%% file: techmain.tex
\documentclass[letterpaper,11pt]{article}
\pdfoutput=1
\newif\ifdraft
\draftfalse
\usepackage[utf8]{inputenc}
\usepackage{amsmath,amsthm,amssymb}
\usepackage[margin=1in]{geometry}
\usepackage{xcolor,xspace}
\usepackage{paralist}
\usepackage{enumitem}
\usepackage[noend]{algorithmic}
\usepackage{algorithm}

\definecolor{darkgreen}{rgb}{0,0.5,0}
\definecolor{darkblue}{rgb}{0.1,0.1,0.6}
\usepackage{hyperref}
\hypersetup{
    unicode=false,          
    colorlinks=true,        
    linkcolor=darkblue,          
    citecolor=darkgreen,        
    filecolor=magenta,      
    urlcolor=cyan           
}
\usepackage[capitalize, nameinlink]{cleveref}
\Crefname{algorithm}{Alg.}{Algs.}

\newtheorem{theorem}{Theorem}[section]
\newtheorem{definition}[theorem]{Definition}
\newtheorem{lemma}[theorem]{Lemma}
\newtheorem{claim}{Claim}
\newtheorem*{claim*}{Claim}
\newtheorem{remark}[theorem]{Remark}
\newtheorem{observation}[theorem]{Observation}

\newcommand{\LOCAL}{$\mathsf{LOCAL}$\xspace}
\newcommand{\CONGEST}{$\mathsf{CONGEST}$\xspace}

\newcommand{\eps}{\varepsilon}
\newcommand{\logstar}{\ensuremath{\log^*}}

\newcommand{\RCT}{\textsc{RandomColorTrial}}
\newcommand{\rct}{\RCT}
\newcommand{\smallComp}{\ensuremath{L}}

\newcommand{\trycolor}{\textsc{TryColor}}
\newcommand{\synchronizedcolortrial}{\textsc{SynchronizedColorTrial}}
\newcommand{\colordensenodes}{\textsc{ColorDenseNodes}}
\newcommand{\colorsmalldegreenodes}{\textsc{ColorSmallDegreeNodes}}
\newcommand{\computeacd}{\textsc{ComputeACD}}
\newcommand{\computecliqueoverlay}{\textsc{ComputeCliqueOverlay}}
\newcommand{\colorsparsenodes}{\textsc{ColorSparseNodes}}
\newcommand{\slackgeneration}{\textsc{SlackGeneration}}

\DeclareMathOperator{\poly}{poly}
\DeclareMathOperator{\polylog}{poly\log}
\DeclareMathOperator{\polyloglog}{poly\log\log}

\newcommand{\E}{\mathbb{E}}

\newcommand{\myemail}[1]{\,$\cdot$\, {\small #1}}
\newcommand{\myaff}[1]{\,$\cdot$\, {\small #1}\par\smallskip}

\newenvironment{mycover}
{\list{}{\listparindent 0pt
        \itemindent    \listparindent
        \leftmargin    0.8cm
        \rightmargin   0.8cm
        \parsep        0pt}%
    \raggedright
    \item\relax}
{\endlist}

\begin{document}

\thispagestyle{empty}

\begin{mycover}
{\huge\bfseries\boldmath Efficient Randomized Distributed Coloring in CONGEST \par}
\bigskip
\bigskip
\bigskip

\textbf{Magn\'us M. Halld\'orsson}
\myemail{mmh@ru.is}
\myaff{Reykjavik University, Iceland}

\textbf{Fabian Kuhn}
\myemail{kuhn@cs.uni-freiburg.de}
\myaff{University of Freiburg, Germany}

\textbf{Yannic Maus}
\myemail{yannic.maus@cs.technion.ac.il}
\myaff{Technion, Israel}

\textbf{Tigran Tonoyan}
\myemail{ttonoyan@gmail.com}
\myaff{Technion, Israel}

\end{mycover}

\vspace*{1cm}

\begin{abstract}
  Distributed vertex coloring is one of the classic problems and
  probably also the most widely studied problems in the area of distributed
  graph algorithms. We present a new randomized distributed vertex
  coloring algorithm for the standard \CONGEST model, where the
  network is modeled as an $n$-node graph $G$, and where the nodes of $G$
  operate in synchronous communication rounds in which they can
  exchange $O(\log n)$-bit messages over all the edges of $G$. For
  graphs with maximum degree $\Delta$, we show that the
  $(\Delta+1)$-list coloring problem (and therefore also the standard
  $(\Delta+1)$-coloring problem) can be solved in
  $O(\log^5\log n)$ rounds. Previously such a result was only
  known for the significantly more powerful \LOCAL model, where in
  each round, neighboring nodes can exchange messages of arbitrary
  size. The best previous $(\Delta+1)$-coloring algorithm in the
  \CONGEST model had a running time of $O(\log\Delta + \log^6\log n)$
  rounds. As a function of $n$ alone, the best previous algorithm
  therefore had a round complexity of $O(\log n)$, which is a bound
  that can also be achieved by a na\"{i}ve folklore algorithm. For
  large maximum degree $\Delta$, our algorithm hence is an
  exponential improvement over the previous state of the art.
\end{abstract}

\clearpage
\thispagestyle{empty}

\tableofcontents
\clearpage

\setcounter{page}{1}

\input{d1col-introduction}

\input{algorithm}

\paragraph{Acknowledgements.} This project was supported by the European Union's Horizon 2020 Research and  Innovation Programme under grant agreement 755839
and by the Icelandic Research Fund grants 174484 and 217965.

\bibliographystyle{alpha}
\bibliography{references}

\appendix
\input{app-bounds}

\end{document}

%% file: d1col-introduction.tex
\section{Introduction}
In the distributed vertex coloring problem, we are given an $n$-node network graph $G=(V,E)$, and the goal is to properly color the nodes of $G$ by a distributed algorithm: The nodes of $G$ are autonomous agents that interact by exchanging messages with their neighbors in synchronous communication rounds. At the end, every node needs to output its color in the computed vertex coloring. The standard version of the problem asks for a coloring with $\Delta+1$ colors, where $\Delta$ is the largest degree of $G$, such that the objective is to match what can be achieved by a simple sequential greedy algorithm. Distributed coloring has been intensively studied for over 30 years. The problem has been used as a prototypical example to study distributed symmetry breaking in graphs, and it certainly is at the very core of the general area of distributed graph algorithms, e.g., \cite{barenboimelkin_book}.

\vspace*{-3mm}

\paragraph{Distributed Coloring, State of the Art.}
The first paper to explicitly study the distributed coloring problem was a seminal paper by Linial~\cite{linial87}, which effectively also started the area of distributed graph algorithms. Already then, it was known that by using simple randomized algorithms for the parallel setting~\cite{alon86,luby86}, \emph{with randomization}, the distributed $(\Delta+1)$-coloring problem can be solved in only $O(\log n)$ communication rounds. In fact, even one of the simplest conceivable randomized distributed coloring algorithms solves the problem in $O(\log n)$ rounds~\cite{johansson99}. The algorithm always maintains a partial proper coloring and operates in $O(\log n)$ synchronous phases. In each phase, each uncolored node chooses a uniformly random color among the colors not already used by some neighbor. A simple analysis  shows that each uncolored node can keep its random color with a constant probability, which leads to the $O(\log n)$ runtime bound. This most basic random coloring step will also play an important role in our paper and we will therefore refer to it as {\RCT} in the following. 

Classically, distributed coloring was studied in a variant of the message passing model known as the \LOCAL model, where in each round, nodes are allowed to exchange messages of arbitrary size. Over the years, the main challenges have been to understand the deterministic complexity of the $(\Delta+1)$-coloring problem (e.g., \cite{barenboimelkin_book,barenboimE10,BarenboimEK14,Barenboim16,FHK,BEG18,Kuhn20}) and to understand to what extent $o(\log n)$-time randomized distributed coloring algorithms exist. In fact, these two questions are actually closely related~\cite{chang16exponential}. In a recent breakthrough, Rozho\v{n} and Ghaffari~\cite{RG19} showed that $(\Delta+1)$-coloring (and many other important problems~\cite{SLOCAL17,FOCS18-derand}) can deterministically be solved in $\polylog(n)$ time. Combined with the astonishing recent progress on randomized 
algorithms~\cite{BEPSv3,EPS15,HSS18,CLP20}, this in particular gives randomized $\polyloglog(n)$-time algorithms,
with the best complexity known being $O(\log^5\log n)$~\cite{GGR20}.

With the complexity of distributed coloring in the powerful \LOCAL model being quite well understood, it might now become within reach to also understand the complexity in the much more realistic \CONGEST model, where in each round, every node is only allowed to exchange $O(\log n)$ bits with each of its neighbors.
Many early distributed coloring algorithms work directly in the more restricted \CONGEST model,
but the recent highly efficient randomized algorithms of \cite{EPS15,HSS18,CLP20} unfortunately make quite heavy use of the power of the \LOCAL model. It seems unclear whether and to what extent their ideas can be applied in the \CONGEST model. The best randomized $(\Delta+1)$-coloring algorithm known in the \CONGEST model has a round complexity of $O(\log\Delta + \log^6\log n)$~\cite{BEPSv3,Ghaffari2019,GGR20}. Note that as a function of the number $n$ of nodes  alone, this algorithm still has a running time of $O(\log n)$, which is no faster than the simple 30 years old methods.
Given all the recent progress on distributed coloring, arguably one of the most important open questions regarding this classic distributed problem is the following.

\smallskip
\begin{center}
  \begin{minipage}{0.9\textwidth}
    \it Is there a randomized algorithm in the \CONGEST model that solves the $(\Delta+1)$-vertex coloring problem in time $o(\log n)$ or even in time $\polyloglog(n)$?
  \end{minipage}
\end{center}

\paragraph{Our Main Contribution.} We answer this question in the affirmative and  give a randomized  $(\Delta+1)$-coloring algorithm in the \CONGEST model, which is \textbf{as fast as} the best known algorithm for the problem in the \LOCAL model. As our main result, we prove the following theorem.

\begin{theorem}[simplified]\label{thm:mainResult}
There is a randomized distributed algorithm, in the \CONGEST model, that solves any given instance of the $(\Delta+1)$-list coloring problem in any $n$-node graph with maximum degree $\Delta$ in \textbf{$O(\log^5\log n)$} rounds, with high probability. 
\end{theorem}

Note that our algorithm even works for the more general $(\Delta+1)$-list coloring problem, where every node initially is given an arbitrary list of $\Delta+1$ colors, and the objective is to find a proper vertex coloring such that each node is colored with one of the colors from its list.

Our algorithm follows the paradigm of breaking the graph into sparse and dense parts and processing them separately, which has been the only successful approach for sublogarithmic complexity in the \LOCAL model~\cite{HSS18,CLP20}. By working in the much more restricted \CONGEST model, however, we are forced to develop  general techniques  based on more basic principles.
We show that, under some conditions, the progress guarantee of {\rct} is exponentially better than suggested by its basic analysis. Our analysis extends to a general class of random coloring algorithms 
akin to {\rct}. For coloring dense parts, however, this has to be combined with additional techniques to deal with the major challenge of congestion.

In the following, we first give a high level overview over what is known about randomized coloring algorithms in the \LOCAL model, then briefly discuss the state of the art in the \CONGEST model. 
In \Cref{sec:technicaloverview}, we  overview  existing techniques that are relevant to our algorithm, explain in more detail why it is challenging to use existing ideas in the \CONGEST model, discuss how we overcome the major challenges, and summarize the algorithm and the technical ideas of the paper.
The core technical part of the paper starts with and is outlined in \Cref{sec:main}.

\paragraph{History of Randomized Coloring in the \LOCAL Model.}
The first improvement over the simple $O(\log n)$-time algorithms of \cite{alon86,luby86,johansson99} appeared in \cite{SW10}, where the authors show that by trying several colors in parallel, the $(\Delta+1)$-coloring problem can be solved in $O(\log\Delta +\sqrt{\log n})$ rounds. A similar result was previously proven for coloring with $O(\Delta)$ colors in \cite{KSOS06}.
Subsequently, the \emph{graph shattering technique}, first developed for constructive Lov\'{a}sz Local Lemma algorithms~\cite{beck1991algorithmic}, was introduced by \cite{BEPSv3} to the area of distributed graph algorithms.
Since each node is colored with constant probability in each iteration of {\RCT}, 
$O(\log \Delta)$ iterations suffice to make the probability of a given node remaining uncolored polynomially small in $\Delta$.
This ensures that afterwards, all remaining connected components of uncolored nodes are of $\polylog n$ size, as shown by \cite{BEPSv3}, and they are then typically colored by a deterministic algorithm. 
The deterministic complexity of coloring $N$-node graphs in the \LOCAL model was at the time $2^{O(\sqrt{\log N})}$~\cite{panconesi1992improved}, but has recently been improved to $O(\log^5 N)$~\cite{RG19,GGR20}.
As a result, the time complexity of \cite{BEPSv3} improved from $O(\log\Delta)+2^{O(\sqrt{\log\log n})}$ to $O(\log\Delta+\log^5 \log n)$. We remark that it  was shown in \cite{chang16exponential} that for distributed coloring and related problems, the randomized complexity on graphs of size $n$ is lower bounded by the deterministic complexity on graphs of size $\sqrt{\log n}$. Hence, in some sense, the graph shattering technique 
is necessary.

All further improvements on randomized distributed coloring concentrated on the ``preshattering'' part, i.e., on coloring each node with probability $1-1/\poly(\Delta)$ so that the uncolored nodes form components of size $\polylog n$. An important step towards a sublogarithmic preshattering phase was done by Elkin, Pettie, and Su~\cite{EPS15}, coloring graphs satisfying
a specific local sparsity property. Following this, Harris, Schneider, and Su~\cite{HSS18} achieved preshattering in time $O(\sqrt{\log\Delta}+\log\log n)$, resulting in a $(\Delta+1)$-coloring algorithm with time complexity $O(\sqrt{\log n})$ (in terms of $n$ alone). The algorithm is based on a decomposition of the graph into locally sparse nodes to which the algorithm of \cite{EPS15} can be applied and into dense components that have a constant diameter such that computations within these components can be carried out in a centralized brute-force manner in the \LOCAL model.
Finally, Chang, Li, and Pettie~\cite{CLP20} gave a hierarchical version of the algorithm of \cite{HSS18}, bringing the preshattering complexity all the way down to $O(\log^* n)$. This leads to the current best randomized $(\Delta+1)$-coloring algorithm known with time complexity $O(\log^5\log n)$.

\paragraph{State of the Art in the \CONGEST Model.} While the simple \RCT{} algorithm and thus the $O(\log\Delta)$ preshattering phase of \cite{BEPSv3} clearly also work in the \CONGEST model, the remaining discussed randomized coloring algorithms all use the additional power of the \LOCAL model to various extents, and if one aims to achieve similar results for the \CONGEST model, one faces a number of challenges. 
The fastest known deterministic algorithm to apply on  $\polylog(n)$-size components is based on  decomposing the graph into clusters of small diameter (as defined in \cite{awerbuch89}), where in the \LOCAL model, each cluster can be colored
in a brute-force way, by collecting its topology at a single node. Luckily, this issue has already been solved: The fastest known network decomposition algorithms of \cite{RG19,GGR20} can directly be applied in the \CONGEST model, and by using techniques developed in \cite{censor2017derandomizing,Ghaffari2019,BKM19,HKMN20}, one can efficiently solve the coloring problem in each cluster in the \CONGEST model. This implies that the $(\Delta+1)$-coloring problem can be solved in $O(\log\Delta)+\polyloglog(n)$ rounds. In order to obtain a sublogarithmic-time $(\Delta+1)$-coloring algorithm in the \CONGEST model, the challenge therefore is to develop an efficient preshattering algorithm. In \Cref{sec:technicaloverview}, we discuss how the preshattering techniques used by \cite{HSS18,CLP20} exploit the power of the \LOCAL model, we describe the challenges that arise for the \CONGEST model and also how we tackle those challenges.

\section{Technical Overview }
\label{sec:technicaloverview}

Although extremely simple, {\RCT} lies at the heart of most known efficient randomized coloring algorithms.
Each uncolored node $v$ picks a uniformly random color $c$ from its current \emph{palette} -- colors from its list $\Psi(v)$ that are not already used by its neighbors -- and then executes {\trycolor}($v,c$) (\Cref{alg:trycolor}). 
Barenboim, Elkin, Pettie and Su \cite{BEPSv3} show that w.h.p., the \emph{uncolored degree} (i.e., the degree of the subgraph induced by uncolored nodes) of every vertex of degree $\Omega(\log n)$ goes down by a \textit{constant factor}
in each iteration of {\RCT}, and within $O(\log \Delta)$ steps, each node has degree $O(\log n)$. 
\begin{algorithm}[H]\caption{{\trycolor} (vertex $v$, color $c_v$)}\label{alg:trycolor}
\begin{algorithmic}[1]
\STATE Send $c_v$ to $N(v)$, receive the set $T=\{c_u : u\in N(v)\}$.
\STATE{\textbf{if}} $c_v\notin T$ \textbf{then} permanently color $v$ with $c_v$.
\STATE Send/receive permanent colors, and remove the received ones from $\Psi(v)$.
\end{algorithmic}
\end{algorithm}

Our main technical contribution is a \emph{much faster} algorithm that partially colors the input graph so that the uncolored degree reduces to $\poly\log n$. Once the degree is reduced, one can rely on efficient $\poly\log\log n$-round algorithms: such algorithms are well known in the \LOCAL model \cite{BEPSv3,RG19,GGR20}, and we discuss their \CONGEST counterparts at the end of this section. 
Our degree reduction follows the by-now standard approach  of partitioning the graph into \emph{sparse} and \emph{dense} parts to be colored separately, with most of the action in the dense part. To our knowledge, this approach  is the only one known to have led to sublogarithmic algorithms, even in \LOCAL~\cite{HSS18,CLP20}. (It has also been useful in other models of computing, e.g.,  sublinear algorithms and Massively Parallel Computing~\cite{ParterSu,CFGUZ19,ACK19}.)

  Algorithms in the \LOCAL model profit from the fact that nodes can try many colors simultaneously, or can simulate any centralized algorithm in a small diameter dense part by collecting all the input to one node. These trivialities in \LOCAL become major challenges in \CONGEST. In the remainder of this section, we  describe how we partition the graph and process sparse and dense parts in \CONGEST.

\paragraph{Almost-Clique Decompositions.} 
Inspired by Reed~\cite{Reed98}, Harris, Schneider, and Su \cite{HSS18} introduced an important structure, the \emph{almost-clique decomposition} (ACD). In the variation we use, it is a partition of $V$ into $V_{sparse}, C_1, C_2, \ldots, C_k$, where $V_{sparse}$ contains sparse nodes, while each $C_i$ is an \emph{almost-clique}
with the property that each node in $C_i$ is adjacent to at least $(1-\epsilon)\Delta$ other nodes in $C_i$, $C_i$ contains at most $(1+\epsilon)\Delta$ nodes, and has diameter at most 2
for a given parameter $\epsilon<1$.
The ACD has been used widely in coloring algorithms in different models \cite{ParterSu,CFGUZ19,ACK19,HKMN20,AlonAssadi}.
An ACD can be computed in $O(1)$ rounds in the \LOCAL model rather easily, since, roughly speaking, pairs of nodes can deduce from their common neighborhood size whether they belong to the same almost-clique or not.  
Assadi, Chen and Khanna \cite{ACK19} gave a method that was used in \cite{HKMN20} to  compute an ACD in $O(\log n)$ rounds in \CONGEST. Their argument is based on random sampling and requires $\Omega(\log n)$ time. 

We overcome this \textbf{first  obstacle} towards a $o(\log n)$ \CONGEST algorithm by showing that such a decomposition can be computed in $O(1)$ rounds via a bootstrapping-style procedure (see \Cref{ssec:acdComputation}). 
 We begin by sampling a subset $S$ of vertices, where each vertex is independently sampled w.p.\ $1/\sqrt{\Delta}$. The idea is to use the whole graph to relay information between the $S$-nodes, in order to efficiently find out which pairs of $S$-nodes belong to the same almost-clique (to be constructed).   To this end, each node $v\in V$ chooses a random $S$-neighbor (ID) and broadcasts it to its neighbors. A key observation is that if two $S$-nodes  share many neighbors,  they will likely receive each other's IDs many times, which allows them to detect  similarity of their neighborhoods. The rest of the nodes join the structure created by the $S$-nodes to form the partition.

\smallskip

After computing the ACD, our algorithm can then color the sparse nodes in $V_{sparse}$ and then the dense nodes in $V_{dense}=C_1\cup\dots\cup C_k$ separately.

\paragraph{Coloring Sparse Nodes and Slack Generation.} 
Schneider and Wattenhofer \cite{SW10} showed that if there are enough colors in the palettes, the coloring can be extremely quick: e.g., $O(\Delta + \log^{1.1} n)$-coloring in $O(\log^* \Delta)$ rounds of \LOCAL. In this case, the nodes have plenty of \emph{slack}: the difference $S$ between their degree and their palette size is then $\Omega(\Delta)$.
Slack is a property that never decreases (but can increase, which only makes the problem easier). 
Suppose we have slack $S \ge \delta \Delta$, while degrees are clearly bounded by $D \le \Delta$.
After $O(\log (1/\delta))$ iterations of {\RCT}, the nodes have degree at most $D \le S/2$.
From then on, each node can try $D/(2S)$ random colors in parallel, which has a failure probability of only $\exp(-D/(2S))$. 
Thus, the degrees go down by an \emph{exponential factor}, $\exp(D/(2S))$, and since the slack is unchanged, the ratio $D/S$ increases as a tower function, resulting in $O(\log^* \Delta)$ time complexity.

Elkin, Pettie, and Su \cite{EPS15} showed that a single execution of {\RCT}, combined with node sampling, 
actually generates sufficient slack in \emph{sparse graphs}. More precisely, if the induced neighborhood graph $G[N(v)]$ of a node $v$ contains $(1-\tau)\binom{\Delta}{2}$ edges, for some $\tau=\Omega(\log n/\Delta)$, then after this {\slackgeneration} step, node $v$ has slack $\Omega(\tau\Delta)$.
The reason is that two non-adjacent neighbors of $v$ have a good chance of being colored with the same color, increasing the slack: the palette size goes down by only one, while the degree goes down by two.
The coloring of the sparse nodes in \cite{HSS18,CLP20} relies on both ingredients: First, a single execution of {\RCT} produces slack of $\Omega(\eps^2 \Delta)$; then the unlimited bandwidth of the \LOCAL model is used to run a version of the multiple color trials \cite{SW10,EPS15,CLP20} described above with $\delta=\Theta(\eps^2)$ in $O(\log (1/\eps) + \log^* \Delta)$ rounds.   This works if $\Delta=\Omega((\log^2 n)/\eps)$, and otherwise, the graph can be colored in  $O(\log (1/\eps))+\poly\log\log n$ rounds \cite{BEPSv3,GGR20}.

The \textbf{second obstacle} is that the multiple trials of \cite{SW10} require large bandwidth and cannot be performed in \CONGEST.
We overcome this by showing that the uncolored degree can be reduced 
to $O(\log n)$ via $O(\log\log \Delta)$ iterations of the very simple {\RCT}.
The probability of vertex $v$ remaining uncolored after iteration $i$ is about $d_i(v)/|\Psi_i(v)|$, where $d_i(v)$ is its degree, and $\Psi_i(v)$ is its palette. With a global upper bound $D_i$ on $d_i(v)$ and a global lower bound $S$ on the slack as a proxy for palette size, this becomes at most $D_i/S$. In the next iteration $i+1$, $D_{i+1} \le D_i^2/S$; hence, the ratio $D_i/S$ satisfies the recurrence $D_{i+1}/S \le (D_i/S)^2$. 
Thus, the uncolored degree goes down very quickly, and in only $O(\log\log\Delta)$ steps, we are left with a low-degree ($O(\log n)$) graph. 

\medskip

\noindent \textbf{\Cref{lem:basiconeshot}} (fast degree reduction, simplified, informal)\textbf{.}  
 \emph{Suppose after each iteration $i$ of a coloring algorithm, every node remains uncolored with probability at most $D_i/S$, even if the random bits of other nodes are adversarial, for some $D_i = \Omega(\log n)$ upper bounding the uncolored degree in that iteration and $S \ge 2 D_0$.
 Then, the series $\{\log(S/D_i\}_i$ grows geometrically. In particular, after $O(\log\log \Delta)$ rounds, the uncolored degree becomes $O(\log n)$. }

\medskip

This powerful observation (in a more general form) will also be crucial for coloring dense nodes.

\paragraph{Coloring Dense Nodes.}
Harris, Schneider and Su \cite{HSS18} use the small diameter property of each almost-clique $C$ to coordinate the coloring choices within $C$.
A single leader node gathers the palettes of all nodes in $C$ and then simulates a sequential version of {\rct} on a \emph{random ordering} of $C$ to ensure that the nodes within $C$ choose different colors. This has the advantage that only \emph{external neighbors} of a node $v$, i.e., the neighbors outside of $C$,  conflict with $v$'s choice, reducing the failure probability to $e(v)/|\Psi(v)|$, where $e(v)$ denotes the \emph{external degree} of $v$. By the ACD definition, this external degree is initially at most $\epsilon \Delta$, while the palette size $|\Psi(v)|$ is $\sim \Delta$, implying a probability to remain uncolored of $O(\epsilon)$. In order to reduce the uncolored degree, they repeat these synchronized color trials. Their main effort is to show that the ratio $E/D$ stays (not much worse than) $\epsilon$ in each repetition, where $E$ is a global upper bound on external degrees and $D$ lower bounds the uncolored degree as and thus palette size. 
Thus, if $\epsilon$ is chosen small (subconstant), the dense nodes are colored fast, while if $\epsilon$ is large, the sparse nodes in $V_{sparse}$ are colored fast. The best tradeoff is found for $\epsilon = \exp(-\Theta(\sqrt{\log \Delta}))$, which yields a  time complexity of $O(\sqrt{\log \Delta})+\poly\log\log n=O(\sqrt{\log n})$ for coloring sparse and dense nodes.

We use a similar synchronized version of {\rct} as \cite{HSS18} but we allow the leader  to process the nodes in an arbitrary order. The crux of it is shown as Alg.~\ref{alg:synchtrial2}, where $X \subseteq V_{dense}$ is a subset of uncolored dense nodes that we apply it on (see later), and $X^C=X\cap C$ is its share in each almost-clique. At a high level, our \CONGEST algorithm has a similar overall structure, with synchronized color selection of nodes within each almost-clique, while the analysis is quite different.
\begin{algorithm}[H]\caption{{\synchronizedcolortrial} (\LOCAL version, informal)}
\label{alg:synchtrial2}
  \begin{algorithmic}[1]
     \STATE Each node $v\in X^C$ sends its palette $\Psi(v)$ to its leader $w_{C}$. 
    \STATE $w_C$ processes the nodes in $X^C$ in an \emph{arbitrary order} $v_1,v_2\dots$, where $v_j$ is assigned a candidate color $c_j$ chosen uniformly at random from $\Psi(v) \setminus \{c_1, c_2, \ldots, c_{j-1}\}$. 
    \STATE $w_C$ sends each node $v_j$ its candidate color $c_j$.
    \STATE {\trycolor}($v_j$, $c_j$), for all $j\ge 1$ in $G[X]$.
\end{algorithmic}
\end{algorithm}

 One crucial difference of our work from~\cite{HSS18} is showing that dense nodes are colored
fast even  with \emph{constant} $\eps$, which effectively eliminates the need to balance the choice of $\eps$ between sparse and dense node coloring.
Initially, we only care about reducing the external degree. To this end, we  focus on the ratio $e(v)/S_v$, where $S_v$ is the slack of node $v$, and we show that it follows the same progression as the ratio $d_v/S_v$ did for the sparse nodes. This builds on the following key structural insight (which only holds for constant $\eps$):

\medskip

\noindent \textbf{\Cref{lem:sparseGetsSlack}}\textbf{.} (simplified) 
 \emph{After {\slackgeneration}, every node $w$ with external degree $e(w) = \Omega(\log n)$  has slack $\Omega(e(w))$.
 }

\medskip

In contrast, \cite{HSS18} does not derive any (initial) slack for dense nodes.

\smallskip

We would now want to apply \Cref{lem:basiconeshot}  to shrink the external degrees in $O(\log\log \Delta)$ iterations of synchronized {\RCT}, but we need to deal first with the heterogeneity in slack and external degrees between different nodes: although our formal version of \Cref{lem:basiconeshot} is more robust to heterogeneous slack, it still requires a non-trivial global lower bound on the slack.  
One way to achieve this would be to partition $C$ into groups of roughly equal slack and color those separately.
Instead, we put aside a subset $C' \subset C$ of  nodes to be colored later, where $|C'| = \Omega(\sqrt{|C|})$.\footnote{We use a more subtle partitioning in the implementation in the \CONGEST model.} 
Each node $v$ in $R_0^C = C - C'$, the set that we color first, has a lower bound $|\Psi(v)| \ge |C'|$ on palette size  throughout the execution of the algorithm. 
We can view this as \emph{effective slack} $ES_v = \Omega(\max(e(v),\sqrt{|C|})$. 
The simplified version of \Cref{lem:basiconeshot} given earlier is extended in the full version to allow for different values among the nodes. It gives a geometric progression in terms of $\log ES_v/e(v)$, so that after after $O(\log\log \Delta)$ iterations, 
the external degree of each node is down to $O(\log n)$, while the effective slack is $\Omega(\sqrt{|C|})$. After three more rounds, each node remains uncolored with probability at most $(O(\log n)/\sqrt{|C|})^3 = O(\log n)/|C|$, resulting in a low-degree subgraph.
We are left with the subgraph $C'$, which we solve recursively, setting aside another subset $C''$ with $|C''|=\Omega(|C|^{1/4})$ and coloring  $R_1^C = C' - C''$ next. We therefore form \emph{layers} $R_1^C, R_2^C, \ldots$ that are colored iteratively,
which adds another $\log\log \Delta$ factor to the runtime. As we shall see, these layers have another advantage that make \CONGEST implementation possible.

We summarize our algorithm for dense nodes in Alg.~\ref{alg:dense2}.
  \begin{algorithm}[H]
\caption{{\colordensenodes} (\LOCAL version, informal)}          
\label{alg:dense2}
\begin{algorithmic}[1]
\STATE Partition the uncolored nodes of each almost-clique $C$ into layers $R^C_0, R^C_1, \ldots, R^C_t$ (TBD) 
\FOR {$i=0,\dots,t-1$}
\STATE{\textbf{for}} $O(1)$ iterations \textbf{do} {\rct} in $R_i$. \label{st:RCTreduce}
\STATE{\textbf{for}} $O(\log\log \Delta)$ iterations \textbf{do} {\synchronizedcolortrial}($R_i$).
\ENDFOR  
\STATE {\colorsmalldegreenodes} in $G[V\setminus V_{sparse}]$.
\end{algorithmic}
\end{algorithm}
We call {\RCT} (in line~\ref{st:RCTreduce}) in order to decrease the external degree $e(v)$ from $O(S_v)$ to at most half the slack $S_v/2$, as needed to apply \Cref{lem:basiconeshot}.

Even if simplified, there still are several daunting challenges in implementing our algorithm in \CONGEST. In particular, \textbf{the main obstacle} is the communication requirements for the leader of each almost-clique to learn the palettes of the nodes, $\Theta(\Delta^2)$ messages in \CONGEST, of which only $m=O(\Delta)$ can be received per round and only when sent by $m$ \emph{distinct} neighbors. 
Our first step towards overcoming this obstacle is  \emph{sparsifying} the palette $\Psi(v)$, by transmitting only a random subset to the leader. This was used earlier in a congested clique algorithm of Parter and Su \cite{ParterSu}, though on a subgraph that was already polynomially smaller.
The presence of the random put-aside subset $C' = \cup_{i\ge 1} R_i^C$ of nodes  means that it suffices  to transmit only a $O(\log n / |C'|)$-fraction of the palette, and yet  {\synchronizedcolortrial} will have enough randomness in the choice of candidate colors of nodes in $R_0^C = C - C'$. In general, when coloring layer $i$, the size of the palette subset to be forwarded is $O(|R_i^C|)\cdot O(\log n )/\sum_{j > i} |R_{j}^C| = O(\log n \cdot |R_i^C|/|R_{i+1}^C|)\ll\Delta$, and the total amount of palette information transmitted to the leader is $O(\log n \cdot |R_i^C|^2/|R_{i+1}^C|)$ colors. We form the layer sizes so that this is always bounded by $O(\Delta)$. 

Note that $O(\Delta)$ colors can be directly sent to the leader only when they come from $O(\Delta)$ distinct  neighbors, which is not the case in our setting: we need to quickly re-route the messages. 
This issue is resolved in a \emph{clique} by a well known routing strategy of Lenzen \cite{Lenzen13} that allows us to satisfy arbitrary communication patterns, as long as each node sends/receives $O(\Delta)$ messages. 
To make Lenzen's result usable in an almost-clique $C$, we compute a \emph{clique overlay} on $C$, that allows for a simulation of all-to-all communication, a \emph{congested clique}, with constant-factor  overhead.  

\medskip

\noindent\textbf{\Cref{thm:congestedclique} }(simplified)\textbf{.} 
\emph{There is a $O(\log\log n)$-round \CONGEST algorithm that for any almost-clique $C$, computes a clique overlay, which simulates all-to-all communication in $C$ with constant-factor runtime overhead.}

\medskip

Since every almost-clique $C$ has diameter 2, constructing such an overlay corresponds to finding a relay node $w$ for each non-edge $uv$ in $G[C]$, in a way that no edge is adjacent to many relays. We reduce this problem to a coloring problem on a graph with vertices corresponding to non-edges $uv$ in $G[C]$, and solve it with similar ideas as the coloring of sparse nodes.

Chang, Li, and Pettie \cite{CLP20} build on the argument of \cite{HSS18} and form a hierarchy of almost clique-decompositions, with epsilons ranging from constant to $\Delta^{-1/10}$. They divide the nodes into layers, further partition them into blocks, which are then grouped together into six separate and different subsets that are tackled with slightly different variations of {\synchronizedcolortrial}. In a tour de force, they show how to reduce them in only a constant number of steps to the setting where their extension of the multi-trials of \cite{SW10} can take over. The argument is quite delicate, as they must also ensure along the way that \emph{not too many} nodes get colored in a round. While it is not impossible \emph{a priori} to implement it in \CONGEST, it is likely to be much more difficult than our approach and not likely to result in a significantly faster method.

The approach of \cite{CLP20} utilizes the unbounded communication bandwidth of the \LOCAL model in an additional way. 
The leader gathers also the topology of the almost-clique, in addition to the palettes of the nodes. We deal with this by conservatively assuming that $C$ is fully connected when it comes to candidate color assignment in {\synchronizedcolortrial}. 
Namely, we bound the \emph{anti-degree} of each node $C$, or its number of \emph{non-neighbors} within $C$, by the slack of the node, as shown in the full version of \Cref{lem:sparseGetsSlack}. The effect of this is roughly equivalent to doubling the external degree, which does not change the fundamentals.

We believe that our approach yields a simpler and more direct way of obtaining a $\poly\log\log (n)$-round algorithm for coloring dense nodes, even in the \LOCAL model.

\paragraph{Coloring Small Degree Graphs \& Postshattering}
\label{ssec:postshattering}

After we have  reduced the uncolored degree to $\Delta'=\poly\log n$, vertices have already lost colors from their initial palette, due to colored neighbors; thus, the remaining problem is a \emph{$(deg+1)$-list coloring problem}, where the nodes can have different degrees, and the list of a node of degree $d$ has size at least $d+1$. 
Hence, even though our general result (\Cref{thm:mainResult}) only solves the  $(\Delta+1)$-list coloring problem, it is essential that the following theorem solves the (possibly harder) $(deg+1)$-list coloring. We use the corresponding algorithm, called {\colorsmalldegreenodes}, to color all remaining vertices in $O(\log\Delta'+\log^5\log n)=O(\log^5\log n)$ rounds.

\smallskip

\textbf{\Cref{thm:smallDegree}} (simplified)\textbf{.} \emph{$(deg+1)$-list coloring  in a graph with $n$ nodes and maximum degree $\Delta'$ can be solved in $O(\log \Delta'+\log^5\log n)$ \CONGEST rounds, w.h.p. }

\smallskip

In the \LOCAL model, an algorithm with the same runtime is known \cite{BEPSv3,RG19,GGR20}, and in the \CONGEST model, a $O(\log \Delta'+\log^6\log n)$-round algorithm is claimed in  \cite{Ghaffari2019,GGR20}. We present a similar (but not identical) algorithm that fixes an error in the design of a subroutine in \cite{Ghaffari2019,GGR20} (see Remark~\ref{rem:smallError}) and  improves the runtime  to match the one in \LOCAL.

Our algorithm is obtained by an improved combination of previously known ideas.
It is based on the shattering framework of \cite{BEPSv3}, which uses $O(\log \Delta')$ iterations of {\rct} to reduce the problem to coloring connected components of \emph{size} $N=\poly\log n$, the \emph{network decomposition} algorithm from \cite{GGR20}, which partitions each such component into clusters of small \emph{diameter} ($\poly\log\log n$), and a \CONGEST algorithm from \cite{Ghaffari2019}, for list coloring a single cluster.
Interestingly, the latter algorithm is also based on \RCT{}. It simulates $O(\log n)$ independent instances in parallel, where each instance runs for $O(\log N)$ iterations. One of these instances is likely to color all vertices of the cluster, and the nodes use the small diameter to agree on a successful instance.

To simulate many instances in parallel, it is essential that many color trials can be sent in one $O(\log n)$-bit message. For this, the nodes in a cluster  compute a mapping of the colors in nodes' palettes to a smaller color space of size $\poly N$ (cf. \cite{HKMN20}): each color then can be represented with $O(\log\log n)$ bits, and  $O(\log n/\log\log n)$ color trials  fit in a single \CONGEST message.  

The bottleneck of our approach is the computation of the network decomposition, which is also the current bottleneck for faster randomized (and deterministic) algorithms in the \LOCAL model.

\section{Top Level Algorithm and Paper Outline}\label{sec:main}

Our main algorithm also serves as an outline of the remainder of the paper.
 \begin{algorithm}[H]
\caption{Main Algorithm Outline}          
\label{alg:main}
\begin{algorithmic}[1]
\STATE{\textbf{if}} $\Delta=O(\log^4 n)$ \textbf{then} {\colorsmalldegreenodes} (Sec.~\ref{sec:smalldegree}) and \textbf{return}
\STATE {\computeacd} Compute an $(\eps,\eta)$-ACD $V_{sparse}, C_1,\ldots, C_k$ of $G$ with $\eps=\frac{1}{3}$, $\eta=\frac{\eps}{108}$ (Sec.~\ref{sec:ACD})\label{st:mainacd}
\STATE {\computecliqueoverlay} for each $C_i$, $1\le i\le k$ (Sec.~\ref{sec:congestedcliqueoverlay})\label{st:overlay}
\STATE Step 1: {\slackgeneration} (Secs.~\ref{sec:oneshot} and~\ref{sec:slack})\label{st:mainslack}
\STATE Step 2:  {\colorsparsenodes} Color remaining sparse nodes in $G$ (Sec.~\ref{sec:sparsecoloring})
\STATE Step 3: {\colordensenodes} Color remaining dense nodes in $G$ (Sec.~\ref{sec:densecoloring})
\end{algorithmic}
\end{algorithm}

The formal statement of our main result is as follows. 

\medskip
\noindent\textbf{Theorem~\ref{thm:mainResult}.}
\emph{
Let $G$ be the input graph with $n$ vertices and maximum degree $\Delta$, where each vertex $v$ has a list $\Psi(v)\subseteq [U]$ of $|\Psi(v)|=\Delta+1$ available colors from a colorspace of size $U=\poly n$. There is a randomized  algorithm that in $O(\log^5\log n)$ rounds in the \CONGEST model, w.h.p.\footnote{With high probability, that is, with probability at least $1-n^c$, for any constant $c\ge 1$.} computes a coloring of $G$  such that each vertex $v$ gets a color from its list $\Psi(v)$.}
\emph{The runtime can also be stated as $O(T+\log\log n+\log^2\log\Delta)$, where 
$T$ is the time needed to $(deg+1)$-list color an $n$-vertex graph with maximum degree $O(\log^4 n)$ in a colorspace of size $\poly n$.}
\begin{proof}[Proof sketch]
If $\Delta=O(\log^4 n)$, then we directly apply algorithm {\colorsmalldegreenodes} from \Cref{sec:smalldegree}, to color $G$ in $T$ rounds. Assume that $\Delta=\omega(\log^4 n)$.
After the ACD computation ($O(1)$ rounds), any node is either in $V_{sparse}$ or in one of the almost-cliques $C_1,\ldots,C_k$; thus, any node that is not colored in {\slackgeneration} ($O(1)$ rounds) gets colored either in {\colorsparsenodes} ($O(\log\log \Delta + T)$ rounds, cf. \Cref{lem:coloringSparse}) or  in {\colordensenodes} ($O(\log\log n +\log^2\log\Delta + T)$ rounds, cf. \Cref{lem:coloringDense}), and, all nodes are colored at the end, w.h.p. 
The overlay computation takes $O(\log\log n)$ rounds. The total runtime is dominated by the last step.
\end{proof}

\section{Definitions and Notation}
We use $n$ and $\Delta$ to denote the number of nodes and the maximum degree of the input graph $G$. We let $d_v$, $\Psi(v)$, and $N(v)$ (or $N_H(v)$, in a subgraph $H$) denote the degree, palette, and the neighborhood, respectively, of a node $v$. We often use $s_v$ for a lower bound on the slack (defined below) of node $v$. For simplicity of exposition, we assume in our analysis that the degree, neighborhood, and palette of a node are always updated, by removing colored neighbors or  colors permanently taken by neighbors from consideration (which can be done in a single  round).

\begin{definition}[similarity, friends, density] ~ Let $\eps>0$ and $G=(V,E)$ be a graph with maximum degree $\Delta$. 
    Nodes $u,v\in V$ are \emph{$\eps$-similar}  if $|N(v)\cap N(u)|\ge (1-\eps)\Delta$, and are \emph{$\eps$-friends}, if in addition $\{u,v\}\in E(G)$. 
    Node $u\in V$ is \emph{$\eps$-dense} if it has $(1-\eps)\Delta$ friends.
\end{definition}

\begin{definition}[Almost-Clique Decomposition (ACD)]\label{def:acd}
Let  $G=(V,E)$ be a graph and $\eps,\eta\in (0,1)$. A partition $V=V_{sparse}\cup C_1\cup \ldots\cup C_k$ of $V$ is an \emph{$(\eps,\eta)$-almost-clique decomposition} for $G$ if:
\begin{compactenum}
    \item $V_{sparse}$ does not contain an $\eta$-dense node\ ,
    \item For every  $i\in [k]$, $(1-\eps)\Delta\le |C_i|\le (1+\eps)\Delta$\ ,
    \item For every $i\in [k]$ and $v\in C_i$,   $|N(v)\cap C_i|\ge (1-\eps)\Delta$\ .
\end{compactenum}
\label{D:acd}
\end{definition}
We refer to $C_i$ as \emph{$\eps$-almost-cliques}, omitting $\eps$ when clear from the context. For a node in an almost clique $C_i$, the \emph{antidegree} -- the number of non-neighbors in $C_i$ --  and \emph{external degree} -- the number of neighbors in other almost cliques $C_j$, $j\neq i$ --  are key parameters that we will use.

\begin{definition}
The \emph{external degree} $e(v)$ of a node $v\in C_i$ is  $e(v)=|N(v)\cap \cup_{j\neq i} C_j|$. The \emph{antidegree} $a(v)$ of a node $v\in C_i$ is $a(v)=|C_i\setminus N(v)|$.
\end{definition}

\Cref{lem:acdproperties}  states useful properties of almost-cliques that easily follow from the ACD definition. In \Cref{sec:ACD}, we show how to compute ACD, for some constants $\eps,\eta$, in $O(1)$ rounds.

\begin{observation}\label{obs:deltaintersections} Let $a,b,c\in (0,1)$.
Let $A,B,C$ be three sets. If $|A\cap C|\ge (1-a)\Delta$, $|B\cap C|\ge (1-b)\Delta$, and $|C|\le (1+c)\Delta$, then $|A\cap B\cap C|\ge (1-a-b-c)\Delta$.
\end{observation}
\begin{lemma}[ACD properties]\label{lem:acdproperties}
Let $C\in \{C_1,\dots,C_k\}$ and $u,v\in C$. It holds that: (i) $u$ and $v$ are $3\eps$-similar, (ii)  $C$ has diameter 1 or 2, (iii) $u$ is $3\eps$-dense, (iv) $e(u)\le \eps\Delta$, and (v) for every $w\in C'\neq C$, $|N(u)\cap N(w)|\leq 2\eps\Delta$.
\end{lemma}
\begin{proof}
By the definition of ACD, $|N(u)\cap C|,|N(v)\cap C|\ge (1-\eps)\Delta$, and $|C|\le (1+\eps)\Delta$; hence, by Obs.~\ref{obs:deltaintersections}, $|N(u)\cap N(v)\cap C|\ge (1-3\eps)\Delta$. This implies (i) and   (ii) (since $u$ and $v$ have a common neighbor in $C$). (iii) follows from (i) and the bound $|N(u)\cap C|\ge (1-\eps)\Delta$. (iv) follows  from  Def.~\ref{D:acd}, 3, using $d_v\le \Delta$. For (v), consider $u\in C$ and $w\in C'$. We have $|N(u)\cap C|,|N(w)\cap C'|\ge (1-\eps)\Delta$, and since $C\cap C'=\emptyset$, it follows that $|N(u)\cap N(w)|<2\eps\Delta$.
\end{proof}

\begin{definition}[slack]
\label{def:slack}
 ~Let $v$ be a node with a color palette $\Psi(v)$ in a subgraph $H$ of $G$. The \emph{slack} of $v$ in $H$ is the difference $|\Psi(v)|-d$, where $d$ is the number of uncolored neighbors of $v$ in $H$.  
\end{definition}

%% file: algorithm.tex
\section{Fast Degree Reduction with RandomColorTrial}\label{sec:oneshot}

This section is devoted to the derivation of the key property of {\rct}  that our algorithms rely on -- fast degree reduction. 
Recall that in {\rct}, every node $v$ picks a color $c_v\in \Psi(v)$ independently, uniformly at random, and calls {\trycolor}($v$, $c_v$) (Alg.~\ref{alg:trycolor}).
We use $N(v)$ to denote the set of uncolored neighbors of node $v$ in $G$.

We begin by observing that, generally, the algorithm provides a constant coloring rate, that is, in a large enough subset, a constant fraction of nodes are colored after a single application. 
\begin{lemma}\label{lem:oneshotslowdecrease}
Let $S$ be a subset of at least $c\log n$ nodes, each node $v\in S$ having palette of size $|\Psi(v)|\ge c\log n$, for a large enough constant $c>0$. After a single application of {\rct}, at least $|S|/9$ nodes in $S$ are colored, w.h.p. 
\end{lemma}
\begin{proof}
It is shown in \cite[Lemma 5.4]{BEPSv3} that under the conditions of our lemma, if $|N(v)|\ge (c/2)\log n$ holds for all nodes in $S$, then at least $|S|/16$ nodes in $S$ are colored in {\rct}, w.p. $1-n^{-c/1024}-n^{-c/32+1}$. To prove our lemma, we partition $S$ into two subsets $S_1$ and $S_2=S\setminus S'$, where $S_1=\{v\in S : |N(v)|\ge |\Psi(v)|\}$.
Let $S'$ be the larger of the two; note that $|S'|\ge (c/2)\log n$. If $S'=S_1$, then by the reasoning above, for a large enough $c$, it holds  w.h.p. that $|S'|/16\ge |S|/32$ nodes in $S$ are colored.  Otherwise, $|\Psi(v)|\ge 2|N(v)|$ holds, for every node $v\in S'$. It is easy to see that in this case, every node in $S$ is successfully colored w.p. at least 1/2, even under an adversarial assignment of colors to its neighbors; hence Chernoff bound (\ref{eq:chernoffless}) implies that w.p. $1-n^{-c/16}$, at least $|S'|/4=|S|/8$ nodes in $S'$ are colored. This completes the proof.
\end{proof}

Next, we show that if nodes have large enough slack, then {\rct} reduces the degrees of nodes very rapidly.
We state the following lemma in a general form, as it is also used in \Cref{sec:densecoloring} for a more involved color trial process. 
\begin{lemma}[fast degree reduction]\label{lem:basiconeshot}
Let $H$ be a subgraph of $G$. For a vertex $v\in H$, let $d_v$ be its degree and $s_v\ge 2d_v$ be a parameter such that $s_v\ge c\log n$, for a large enough constant $c>0$. Assume the nodes in $H$ execute a randomized coloring algorithm where in every round $i$, each uncolored node picks a color different from its neighbors' choices and makes it its permanent color w.p.\ at least $1-|N_i(v)|/s_v$, irrespective of the color choices of other nodes, where $N_i(v)$ is the set of uncolored neighbors of $v$ in $H$ at the beginning of round $i$.
After $O(\log\log s^*)$ iterations, every node $v$ in $H$ has at most $O((s_v/s^*)\log n)$ uncolored neighbors in $H$, where $s^*=\min_w s_w$. 
\end{lemma}
\begin{proof}  
Let $n_v^i$ be an upper bound (specified below) on the size of $N_i(v)$. Initially, we have $n_v^1=d_v$.  
As we assumed, node $v$ is colored in iteration $i$ w.p.\ at least $1-n_v^i/s_v$, irrespective of the outcome for  other nodes. For each node $u$, let $Z_u=1$ if $u$ is not colored in iteration $i$ and $Z_u=0$ otherwise. We have $Pr[Z_u=1]\le n_u^i/s_u$. 
Consider a node $v$, and let $Z=\sum_{u\in N_i(v)}Z_u$. 
We have 
$
\E[Z]\le \sum_{u\in N_i(v)} n_u^i/s_u\le n_v^i M_i,
$
where $M_i=\max_u (n_u^i/s_u)$.
Since the probability bound on $Z_u$ above holds irrespective of other $Z_{u'}$, we apply Chernoff bound (\ref{eq:chernoffmore}) to get \[
Pr\left[\sum_{u\in N_i(v)}Z_u>n_v^iM_i + c'\log n\right]\le n^{-c'/4}\ ,
\]
for any constant $c'>4$. 
Hence it holds w.h.p.\ that $\sum_{N_i(v)}Z_u\le n_v^iM_i + c'\log n$, and we can set 
$n_v^{i+1}= n_v^iM_i + c'\log n$. Consider a node $w$ with $n_w^i/s_w\ge M_i^{3/2}$. Then we have 
\[
\frac{n_w^{i+1}}{s_w} = \frac{n_w^iM_i}{s_w} + \frac{c'\log n}{s_w}\le M_i^{5/3} + \frac{c'\log n}{s^*}\ .
\]
The latter implies the recursion $M_{i+1}\le M_i^{3/2} + c'\log n/s^*$. Let $i_0=O(1)$ be the first index where the recursion holds and  $r=c'\log n/s^*$. Recall that $M_{i_0}\le 1/4$.
Since $M_i,r< 1$, we have:
\[
M_{i+1}\le M_i^{3/2} + r\le M_{i-1}^{(3/2)^2} + r^{3/2} + r\le M_{i_0}^{(3/2)^{i-O(1)}} + \sum_{j=0}^{i} r^{(3/2)^j}=M_{i_0}^{(3/2)^{i-O(1)}} + \frac{O(\log n)}{s^*}\ .
\]
Thus, after $i=O(\log\log s^*)$ iterations we have $M_{i}=\frac{O(\log n)}{s^*}$, i.e., $n_v^i= \frac{O(s_v\log n)}{s^*}$, w.h.p. 
\end{proof}

\section{Step 1: Initial Slack Generation}\label{sec:slack}

Slack generation is based on the following idea. Take a node $v$, and let $u,w$ be two of its neighbors. Let $u,w$ each choose a random color from its palette. If either $u$ or $v$ has a palette that significantly differs from $\Psi(v)$, then it, say $u$, is likely to choose a color not in $\Psi(v)$, and the slack of $v$ gets increased, if $u$ retains its color. Otherwise, both $u$ and $v$ have many common colors with $v$, and so are likely to both select the same color, again increasing the slack of $v$, if they retain their colors. This increase of slack can be attributed to the pair $u,w$. It is possible to also ensure that nodes are likely to retain the color they pick, by combining random coloring with node  sampling. 
 This leads us to the algorithm called {\slackgeneration}, which consists of an execution of {\rct} in a randomly sampled subgraph $G[S]$. Each node independently joins $S$ w.p.\ $p=1/20$. 

The following key result, which  has appeared in various reincarnations in~\cite{EPS15,HSS18,CLP20,AlonAssadi}, formalizes the intuition above, showing that {\slackgeneration} converts sparsity into slack. We note that none of the known variants is suitable to our setting: either they are not adapted for list coloring or the dependence on sparsity is sub-optimal. 

The \emph{local sparsity} $\zeta_v$ of a vertex $v$ is defined as $\zeta_v=\frac{1}{\Delta}\left({\Delta \choose 2}-m(N(v))\right)$, where for a set $X$ of vertices, $m(X)$ denotes the number of edges in the subgraph $G[X]$. Roughly speaking, $\zeta_v$ is proportional to the number of ``missing'' edges in the neighborhood of $v$; note that $0\le \zeta_v< \Delta/2$. 

\begin{lemma}\label{lem:sparseImpliesSlack} Assume each vertex $v$ has a palette $\Psi(v)$ of size $\Delta+1$. Let $v$ be a vertex
with $\zeta_v\ge c\log n$, for a large enough constant $c>0$.  There is a constant $c_1>0$ such that after {\slackgeneration}, $v$ has slack $Z \ge  c_1\zeta_v$, w.h.p. 
\end{lemma}
\begin{proof}
Let $\zeta=\zeta_v$. We may assume w.l.o.g. that $|N(v)|\ge \Delta-\zeta/2\ge 3\Delta/4$, as otherwise $v$ has slack $\zeta/2$, and we are done. Let $X\subseteq {\binom{N(v)}{2}}$ be the set of pairs $\{u,w\}$ s.t. $\{u,w\}$ is not an edge. Under our assumption, $N(v)$ contains ${\Delta-\zeta/2\choose 2}$ pairs of nodes, of which at most ${\Delta\choose 2}-\zeta\Delta$ are edges, hence $|X|\ge {\Delta-\zeta/2\choose 2}-{\Delta\choose 2}+\zeta\Delta>\zeta\Delta/2$. 

A node $u$ is \emph{activated} if it is sampled in $S$. Note that every node is independently activated w.p.\ $p=1/20$. The rest of the proof is conditioned on the high probability event that for every node $v$, there are at most $(4/3)p\Delta=\Delta/15$ nodes activated in $N(v)$ (which easily follows by an application of Chernoff bound (\ref{eq:chernoffmore})).  We use $V(X)=\{u : \exists w, \{u,w\}\in X\}$ to denote the vertices that are in at least one pair in $X$. Let $X_1\subseteq X$ be the subset of pairs $\{u,w\}$, s.t. at least one $z\in \{u,w\}$ satisfies 
\begin{equation}\label{eq:lowinter}
|\Psi(z)\cap\Psi(v)|<(9/10)(\Delta+1)\ ,
\end{equation}
and let $X_2=X\setminus X_1$. The case $|X_1|\ge |X|/2$ is easier to handle. Let $V_1\subseteq V$ be the set of vertices that are present in a pair in $X_1$ and satisfy (\ref{eq:lowinter}). Since $|X_1|\ge |X|/2\ge \zeta\Delta/4$ and each $w\in V_1$ can give rise to at most $\Delta$ pairs, $|V_1|\ge \zeta/4$. By a Chernoff bound, there is a subset $V'_1$ of at least $\zeta p/5$ nodes activated in $V_1$, w.h.p.\ (where we use the assumption that $\zeta/\log n$ is large enough). Let $w\in V'_1$. As assumed, $|\Psi(w)\setminus \Psi(v)|> \Delta/10$, and there are at most $\Delta/15$ activated neighbors of $w$. Thus, $w$ chooses a color $c\notin\Psi(v)$ and retains it, even when conditioned on arbitrary color choices of its activated neighbors, w.p.\ at least $1/10-1/15=1/30$. Thus, we can apply 
 Chernoff bound (\ref{eq:chernoffless}), to obtain that at least $p\zeta/160$ nodes in $V_1$ choose a color that is not in $\Psi(v)$, w.h.p.\ (again, we use the assumption that $\zeta/\log n$ is large enough). This gives the claimed slack to $v$.  
 Therefore, we continue with the assumption that $|X_2|\ge |X|/2$, or for simplicity, that $X=X_2$ and $|X|\ge \zeta\Delta/2$. Thus, for every pair $\{u,w\}\in X$, $|\Psi(w)\cap\Psi(v)|\ge (9/10)\Delta$, and $|\Psi(u)\cap\Psi(v)|\ge (9/10)\Delta$, and by Obs.~\ref{obs:deltaintersections},  
\begin{equation}\label{eq:slackgen1}
|\Psi(u)\cap \Psi(w)|\ge (4/5)\Delta\ .
\end{equation}
 Let $Z$ be the number of colors in $\Psi(v)$ that are picked by at least one pair of activated vertices in $X$ and are permanently selected by all these activated neighbors.

For each pair $\{u,w\}\in X$, let $\mathcal{E}_1$ be the event that $u$ and $w$ are activated and pick the same color $c$ which is not picked by any other node in $N(u)\cap N(w)$, and let $\mathcal{E}_2$ be the event that the colors picked by $u$ and $w$ are not picked by any other node in $V(X)$. Let $Y_{u,w}$ be the binary random variable that is the indicator of the event $\mathcal{E}_1\cap \mathcal{E}_2$.

Note that $Z\ge Y=\sum_{\{u,v\}\in X} Y_{u,v}$. First, we show that $Pr[Y_{u,v}=1]=\Omega(1/\Delta)$, which implies that $\E[Z]\ge \E[Y]=\Omega(|X|/\Delta)=\Omega(\zeta)$.

Note that $Pr[Y_{u,v}=1]=Pr[\mathcal{E}_1\cap \mathcal{E}_2]=Pr[\mathcal{E}_1] \cdot Pr[\mathcal{E}_2\mid \mathcal{E}_1]$.  Note that  $Pr[\mathcal{E}_2\mid \mathcal{E}_1]$ is the probability that no node in $V(X)\setminus (\{u,w\}\cup N(u)\cup N(w))$ picks the color $c$ picked by $u$ and $w$. Since each node independently is activated w.p.\ $p$, the probability that a node picks $c$ is $1/(\Delta+1)$, and $|V(X)|\le\Delta$, we have
\[
Pr[\mathcal{E}_2\mid \mathcal{E}_1]\ge (1-p/(\Delta+1))^{\Delta}=\Omega(1)\ .
\]
Next, consider $Pr[\mathcal{E}_1]$. By our assumption above, for every pair of vertices $u,w$, at most $3p\Delta<\Delta/5$ vertices are activated in $N(u)\cup N(w)$. By (\ref{eq:slackgen1}), $|\Psi(u)\cap\Psi(w)|\ge (4/5)\Delta$; hence, there is a set $S_{u,w}$ of at least $(3/5)\Delta$ colors in $\Psi(u)\cap \Psi(w)$ that are not selected by a neighbor in $N(u)\cap N(w)$. Note that $\mathcal{E}_1$ is implied by both $u$ and $w$ being activated and choosing a color from $S_{u,w}$, hence we have, as claimed,
\[
Pr[\mathcal{E}_1]\ge p^2\cdot \frac{|S_{u,w}|}{|\Psi(u)|}\cdot \frac{1}{|\Psi(w)|}=\Omega(1/\Delta)\ .
\] 

It remains to show that $Z$ is concentrated around its mean. Let $T$ be the number of colors in $\Psi(v)$ that are picked by at least one pair $(u,w)\in X$, and $D$  be the number of colors in $\Psi(v)$ that are picked by at least one pair $(u,w)\in X$ but are not retained by at least one of them. Note that $Z=T-D$. Moreover, note that both $T$ and $D$ are a function of activation r.v. and the random color pick of vertices in $V(X)\cup N(V(X))$, and as such they are (i) $\Theta(1)$-certifiable: for  each color picked in $T$, there are 2 nodes in $V(X)$ that ``can certify'' for it (for $D$, an additional one that picked the same color), and (ii) $\Theta(1)$-Lipschitz: changing the color/activation of one vertex can affect $T$ and $D$ by at most 2.

We need to bound $\E[T]$ (which implies the same bound for $D$, as $D\le T$). For a given color $c$ and two nodes $u,w\in V(X)$, the probability that $u$ and $w$ both pick $c$ is at most $1/(\Delta+1)^{2}$. By the union bound, the probability that $c$ is picked by a pair is at most $X/(\Delta+1)^2$. There are $\Delta+1$ colors in $\Psi(v)$, so the expected number $T$ of colors that are picked by at least one pair is $\E[T]\le X/(\Delta+1)<\zeta$, and since $\E[Z]=\Omega(\zeta)=\Omega(\E[T])>c_3\sqrt{\E[T]}$, for a large enough constant $c_3$, since $\zeta=\Omega(\log n)$. Applying  Lemma~\ref{lem:talagrand}, and using these relations, we have
\[
Pr\left[|T-\E[T]|\ge \E[Z]/10\right]\le \exp\left(-\Theta(1)\frac{(\E[Z]/10-O(\sqrt{\E[T]}))^2}{\E[T]}\right)\le \exp(-\Omega(\zeta))\ .
\]
 Since $\zeta\ge c_2\log n$, for a large enough constant $c_2$, we have that $|T-\E[T]|<\E[Z]/10$, w.h.p. Similarly, $|D-\E[D]|<\E[Z]/10$, w.h.p. Putting together we see that w.h.p., $Z=T-D\ge \E[T]-\E[D]-\E[Z]/5=(4/5)\cdot\E[Z]=\Omega(\zeta)$. This completes the proof. 
\end{proof}

We apply Lemma~\ref{lem:sparseImpliesSlack} to establish slack for nodes of various densities. Sparse nodes obtain slack $\Omega(\Delta)$, while more dense nodes have slack depending on the \emph{external degree} and \emph{antidegree}.
We use a connection of antidegree and external degree to local sparsity, originally observed in \cite{HKMN20} for distance-2 coloring. 

\begin{lemma} Let $\eta,\eps\le 1/3$.
For every node $v\in V_{sparse}$, $\zeta_v\ge(\eta^2/4)\Delta$. For every node $v\in C=C_i$ with antidegree $a(v)$ and external degree $e(v)$, it holds that $\zeta_v\ge (1-2\eps)a(v)$ and $\zeta_v\ge (1-3\eps)e(v)/2$.
\end{lemma}
\begin{proof}
Let $v\in V_{sparse}$. Since $v$ is not $\eta$-dense, there are less than $(1-\eta)\Delta$ nodes $u\in N(v)$ satisfying $|N(u)\cap N(v)|\ge (1-\eta)\Delta$. We assume that $d=|N(v)|\ge (1-\eta/5)\Delta$, as otherwise $\zeta_v\ge \eta/10$ follows from the definition of $\zeta_v$. Thus, $v$ has at least $(4\eta/5)\Delta$ neighbors $u$, each having at least $(4\eta/5)\Delta$ non-neighbors in $N(v)$; therefore $m(N(v))\le {\Delta\choose 2}-(8/25)\eta^2\Delta^2$, and $\zeta_v\ge (8\eta^2/25)\Delta\ge (\eta^2/4)\Delta$.

Let us consider a node $v\in C=C_i$. Each external neighbor $w\in N(v)\setminus C$ has at most $\eps\Delta$ neighbors in $C$ (by the definition of ACD). Thus, since $|N(v)\cap C|\ge (1-\eps)\Delta$,  $w$ contributes at least $(1-2\eps)\Delta$ non-edges to $G[N(v)]$. In total, the external neighbors contribute at least $e(v)(1-2\eps)\Delta$ non-edges in $G[N(v)]$, which must be at most $\zeta_v \Delta$; hence, $\zeta_v\ge (1-2\eps)e(v)$.

Next, observe that for every node  $u\in C\setminus N(v)$, $|N(u)\cap N(v)\cap C|\ge (1-3\eps)\Delta$, since $|N(u)\cap C|,|N(v)\cap C|\ge (1-\eps)\Delta$ and $|C|\ge (1+\eps)\Delta$; hence, there are at least $a(v) \cdot (1-3\eps)\Delta$ edges between $N(v)$ and  $C\setminus N(v)$. On the other hand, by the definition of sparsity, at most $2\zeta_v\Delta$ edges can exit $N(v)$. Thus, $\zeta_v\ge (1-3\eps)a(v)/2$.
\end{proof}

The two lemmas above immediately imply the following one.
\begin{lemma} \label{lem:sparseGetsSlack}  
After {\slackgeneration}, every node $v\in V_{sparse}$ has slack $\Omega(\Delta)$, and every node  $w\in C_i$ with antidegree $a(w)$ and external degree $e(w)$ such that $a(w)+e(w)\ge c\log n$, for a large enough constant $c>0$, has slack $\Omega(e(w)+a(w))$, w.h.p.
\end{lemma}

\section{Step 2: Coloring Sparse Nodes}\label{sec:sparsecoloring}

We show that given the $(\eps,\eta)$-ACD, the maximum degree of the graph induced by uncolored nodes in $V_{sparse}$ can be reduced to $O(\log n)$ with $O(\log \log \Delta)$ executions of {\rct} in $G[V_{sparse}]$, after which we can apply {\colorsmalldegreenodes} to finish coloring $G[V_{sparse}]$. We call this procedure {\colorsparsenodes}. 
After slack generation, each node in $V_{sparse}$ has slack $\Omega(\Delta)$ (see Sec.~\ref{sec:slack}).
\Cref{lem:oneshotslowdecrease}  shows that in a few iterations, the slack of every node in $V_{sparse}$ becomes larger than its uncolored degree, after which \Cref{lem:basiconeshot} applies, to show that the degrees of nodes rapidly decrease to $O(\log n)$.
\begin{lemma}[Coloring sparse nodes]
\label{lem:coloringSparse}
After {\slackgeneration}, {\colorsparsenodes} colors all vertices in $V_{sparse}$ in $O(\log \log \Delta + T)$ rounds, w.h.p.,
where $T$ is the time needed to $(deg+1)$-list color an $n$-vertex graph with maximum degree $O(\log^4 n)$ in a colorspace of size $\poly(n)$.
\end{lemma}
\begin{proof} 
Assume that $\Delta\ge c\log n$, for a large enough constant $c>0$, as otherwise $G$ can be colored in $T$ rounds, using {\colorsmalldegreenodes}. After slack generation, each node $v\in V_{sparse}$ has slack  $s_v\ge s^*= c'\Delta$, for a constant $c'>0$ (\Cref{lem:sparseGetsSlack}).
Let $w\in V_{sparse}$ and let $S$ be the set of uncolored neighbors of $w$ in $G[V_{sparse}]$.  
Since each node $v\in S$ has a palette of size $|\Psi(v)|\ge s_v\ge cc'\log n$,  Lemma~\ref{lem:oneshotslowdecrease} implies that after $O(1)$ rounds, there are at most $s_w/2$ uncolored nodes remaining in $S$, w.h.p. By the union bound, this holds for all nodes $w\in V_{sparse}$, w.h.p. 
For the subsequent rounds, we lower bound the probability for a node $v$ to get colored, in order to apply Lemma~\ref{lem:basiconeshot}. In every round $i$, node $v$ picks a uniformly random color from $\Psi(v)$; hence, conditioned on any colors selected by its $t_i$ participating neighbors, $v$ selects a different color w.p.\ at least $(|\Psi(v)| -  t_i)/|\Psi(v)|= 1-t_i/|\Psi(v)|\ge 1-t_i/s_v$, using $s_v\le |\Psi(v)|$. Note that the probability bound on getting colored holds in each iteration as the slack of a node never decreases; hence, we can indeed apply \Cref{lem:basiconeshot} with parameters $s_v=s^*=c'\Delta$ (also recall that $s_v\ge cc'\log n$). The lemma implies that in $O(\log\log \Delta)$ rounds, the maximum degree of $G[V_{sparse}]$ reduces to $O(\log n)$, w.h.p. The remaining nodes in $V_{sparse}$ can be colored in $T$ rounds, using {\colorsmalldegreenodes}.
\end{proof}

\section{Step 3: Coloring Dense Nodes}\label{sec:densecoloring}
We assume here that the $(\eps,\eta)$-ACD $V_{sparse},C_1,\dots,C_k$ of $G$, with $\eps=1/3$ and $\eta=\eps/108$, is given (computed in Alg.~\ref{alg:main}, line~\ref{st:mainacd}), where each almost-clique $C=C_i$ has a designated leader node $w_{C}$ (e.g., the node with minimum ID), as well as a clique overlay (computed in Alg.~\ref{alg:main}, line~\ref{st:overlay}). We further assume that {\slackgeneration} has been executed (line~\ref{st:mainslack}), and Lemma~\ref{lem:sparseGetsSlack} applies. 

In this section we describe algorithm {\colordensenodes} (Alg.~\ref{alg:dense}) that colors the vertices in almost-cliques $C_1,\ldots,C_k$. The high level idea is to use the slack of nodes to reduce the degrees of the subgraph induced by uncolored dense nodes to $\poly \log(n)$, and then color the remaining vertices via {\colorsmalldegreenodes} (\Cref{thm:smallDegree}). We prove the following result.

\begin{lemma}[Coloring dense nodes]
\label{lem:coloringDense}
After {\slackgeneration}, {\colordensenodes} colors all vertices in $C_1,\dots,C_k$ in $O(T+\log \log n + \log^2\log\Delta)$ rounds, w.h.p., where $T$ is the time needed to $(deg+1)$-list color an $n$-vertex graph with maximum degree $O(\log^4 n)$ in a colorspace of size $\poly(n)$.
\end{lemma}
 We present the algorithm in Sec.~\ref{ssec:denseAlg}, prove that all dense nodes are colored in Sec.~\ref{ssec:degreeReduction}, and complete the proof of \Cref{lem:coloringDense} by proving that the algorithm can be implemented efficiently in the {\CONGEST} model in  Sec.~\ref{ssec:randomGreedyImplementation}.
 
 \emph{We assume} that $\Delta=\omega(\log^4 n)$, as otherwise  {\colorsmalldegreenodes} gives the lemma.

\subsection{Algorithm Description}
\label{ssec:denseAlg}
Towards reducing the degrees of nodes, the plan is to first reduce the external degree of each node to $O(\log n)$, and then reduce the size of each almost-clique. If there were no edges within the almost-cliques, one could reduce the external degree of nodes by applying {\rct} $O(\log\log\Delta)$ times, as we did for coloring sparse nodes (\Cref{sec:sparsecoloring}), using the fact that each node has slack proportional to its external degree (by \Cref{lem:sparseGetsSlack}). There are two obstacles to this. First, the external degrees of nodes can be very different, which is problematic when applying the arguments from \Cref{sec:oneshot}. Second, unfortunately, there are many edges in almost-cliques. 

To overcome the first obstacle, we partition each almost-clique into layers of carefully chosen sizes and handle them sequentially.  This ensures that although nodes have different external degrees, they have (to some extent) similar slacks within the subgraph induced by each layer, which is provided by the uncolored neighbors in subsequent layers that are left to be colored later.

To overcome the second obstacle, we assign the nodes in each almost-clique random colors from their palettes \emph{in an arbitrary fixed node order}, ensuring that each node gets a color different from its predecessors (see Alg.~\ref{alg:synchtrial}). A fast implementation of this procedure in the {\CONGEST} model poses certain technical challenges. The idea is to collect the palettes of nodes of an almost-clique into the leader node, which can then choose the candidate colors and send them back to the nodes. This cannot be done quickly, even with fast communication via clique overlays. Instead, we show that it suffices to send a large enough random subset of each palette, and the similar slack of nodes provided by partitioning also comes in handy here. 

After reducing the external degree, a few applications of the color-assignment-and-trial-in-cliques procedure described above suffices to reduce the number of neighbors of each node in the given layer to $\poly \log(n)$. This is achieved due to the internally conflict-free (for each almost-clique) color assignment of node colors, low external degree, and large slack.

 Algorithm {\colordensenodes} is formally described in Alg.~\ref{alg:dense}. 
In line~\ref{st:partition}, we partition each almost-clique $C=C_i$ into $t= O(\log\log\Delta)$ \emph{layers} $R_1^C, \ldots, R_t^C$ that are processed iteratively. The partitioning is done probabilistically, where each vertex independently joins layer $R_i^C$ w.p.\ $p_i$. We describe the probability distribution below. Throughout this section, let $R_i=\cup_{j=1}^kR_i^{C_j}$ denote the set of all vertices of layer $i$. The terms $R_i,R_i^C$ will always denote the corresponding sets of \emph{uncolored} nodes, that is, nodes that are permanently colored are automatically removed from these sets.

  \begin{algorithm}[H]
\caption{{\colordensenodes} }          
\label{alg:dense}
\begin{algorithmic}[1]
\STATE Partition the uncolored nodes of each almost-clique $C$ into layers $R^C_0, R^C_1, \ldots, R^C_t$, where each such node joins $R^C_i$ with probability $p_i$, independently of other nodes. \label{st:partition}
\STATE{\textbf{for}} $O(\log\log n)$ iterations \textbf{do} {\rct} in $R_0$.\label{st:lglgnoneshot}
\STATE {\colorsmalldegreenodes}  in $G[R_0]$.\label{st:smalldeg1}
\FOR {$i=0,\dots,t-1$} \label{st:densemainloop}
\STATE{\textbf{for}} $O(1)$ iterations \textbf{do} {\rct} in $R_i$. \label{st:oneShotFor2}
\STATE{\textbf{for}} $O(\log\log \Delta)$ iterations \textbf{do} {\synchronizedcolortrial} in $R_i$.\label{st:core}
\ENDFOR  
\STATE {\colorsmalldegreenodes} in $G[V\setminus V_{sparse}]$.  \label{st:smalldeg3}
\end{algorithmic}
\end{algorithm}

The purpose of {\rct} in line \ref{st:lglgnoneshot}, followed by {\colorsmalldegreenodes}, is to reduce the size of $R_0^C$ by a logarithmic factor (Lemma~\ref{lem:rcisize}), which is needed for efficient communication using the clique overlay. 

When processing each layer, we start with $O(1)$ applications of {\rct}, with the purpose of increasing the slack by a constant factor (used in \Cref{lem:colorpickprobability}). The main action happens in line~\ref{st:core}, where $O(\log\log\Delta)$ applications of {\synchronizedcolortrial} (described below) reduce the size of layer $i$, so that each node in $V$ has $O(\log^2 n)$ neighbors in $R_i$. After this we can invoke {\colorsmalldegreenodes} to finish coloring $V\setminus V_{sparse}$. The subgraph induced by the last layer $R_t$ has small degree, so it is handled by {\colorsmalldegreenodes} directly.

Alg.~\ref{alg:dense} is executed on all almost-cliques in parallel (in particular, each layer is processed in parallel) and all claims (in particular those in Section~\ref{ssec:degreeReduction}) hold for all almost-cliques.

Finally, let us describe the probability distribution for partitioning. Let $t'=\lceil\log_{3/2}\log_{\log n}\sqrt{\Delta}\rceil$, and $t\le 3\log\log \Delta$, to be specified below. Let $p_1=1/\log^{3/2} n$, and, for $2\le i\le t'$, let $p_i=p_{i-1}^{3/2}$. For $t'<i\le t$, let $p_i=\sqrt{p_{i-1}/\Delta}$, where $t$ is the largest value such that $p_t\ge c\log n/\Delta$, for a sufficiently large constant $c>0$. Finally, let $p_0=1 - \sum_1^t p_i$. 
We let $\Lambda_i=\Delta p_i$ denote  the (roughly) expected size of $R_i^C$. The following observation contains all the properties of the probability distribution that we  need.
\begin{observation}\label{obs:lambdaprops}
Let $n>16$. For every  $c>0$, there is a $c'>0$ such that if $\Delta>c'\log^4 n$, then:

\begin{minipage}{0.5\textwidth}
\begin{enumerate}[label=(\roman*),noitemsep] 
\item  $t$ is well defined, and $t'<t\le 3\log\log\Delta$, 
\item  $\Lambda_0\ge \Delta/4$,
\item  $c\log n \le \Lambda_t\le c^2\log^2 n$ 
\end{enumerate}
\end{minipage}
\begin{minipage}{0.5\textwidth}
\begin{enumerate}[label=(\roman*),noitemsep]
\setcounter{enumi}{3}
\item $\Lambda_i\ge c\log n$,  $0\le i\le t$ 
\item $\Lambda_i\le\Lambda_{i+1}^2$,  $0\le i<t$, and
\item $\Lambda_i\le\sqrt{\frac{\Delta \Lambda_{i+1}}{\log n}}$,  $0<i<t$\label{obsi:sending}
\end{enumerate}
\end{minipage}
\end{observation}
\begin{proof}
\begin{enumerate}[label=(\roman*)]
\item By the definition of $t'$, we have $p_{t'-1}\ge \Delta^{-1/2}$, implying that $p_{t'}\ge \Delta^{-3/4}>c\log n /\Delta$, for $\Delta\ge c^4\log^4 n$. 
Thus $t$ is well defined, we have $t>t'$ and by $t$'s definition we have $p_t\Delta\geq c\log n$. Also note that $p_{t'}\Delta\le \sqrt{\Delta}$. 

For $t\ge i>t'$, we have $p_{t'}\Delta\le \sqrt{\Delta}$, $p_i\Delta=\sqrt{p_{i-1}\Delta}$, and $p_t\Delta\ge c\log n$, hence $t-t'\le \log\log\Delta$. Since $t'\le 2\log\log\Delta$, we obtain (i). 

\item  Note that
for $1\le i\le t'$, we have $p_i=(\log n)^{-(3/2)^i}$, and since $\log n> 4$, it holds that $\sum_{i\le t'} p_i<1/4$. For $i>t'$, we have that $p_i<1/\sqrt{\Delta}$, and since $t-t'\le \log\log\Delta$, we have that $\sum_{i>t'}p_i<\log\log\Delta/\sqrt{\Delta}\le 1/2$, if $\Delta\ge 16$; hence, $p_0=1-\sum_{1}^t p_i\ge 1/4$, and $\Lambda_0\ge \Delta/4$. 

\item  By the definition of $t$, we have $p_t\ge c\log n/\Delta$, and  $\sqrt{p_t/\Delta}<c\log n/\Delta$. They imply that $c\log n\le \Lambda_t< c^2\log^2 n$ 

\item  Follows from the fact that $p_i$ is a decreasing sequence, and $p_t\ge c\log n/\Delta$. 

\item  As observed above, $p_{t'}\ge \Delta^{-3/4}$. 
The latter holds for $i<t'$ as well. Then, $p_{i}\ge \Delta^{-3/4}$ implies that $(p_i\Delta)^{4/3}\ge\Delta^{1/3}$ holds for all $0\leq i\leq t'$. Using this in the last inequality, we obtain, for $t'> i>0$,  
\[
\Lambda_i=p_{i}\Delta =p_{i+1}^{2/3}\Delta= (p_{i+1}\Delta)^{2/3}\cdot \Delta^{1/3}\le (p_{i+1}\Delta)^{2}=\Lambda_{i+1}^2\ .
\]
For $t>i\ge t'$, we have $\Lambda_i=\Lambda_{i+1}^2$, by the definition of $p_i$. For $i=0$, $\Lambda_0/\Lambda_1^2=\log^3/\Delta<1$.

\item For $0<i<t'$, we have
\[
\Lambda_{i}^2=\Delta^2 p_i^2 =\Delta^2 p_{i+1}^{4/3}=\Delta\Lambda_{i+1}\cdot p_{i+1}^{1/3}\le  \Delta\Lambda_{i+1}/\log n
,\]
since $p_{i+1}= \log^{-(3/2)^{i+1}}n\leq \log^{-3} n$ holds for $i\ge 1$.

 For $t>i\ge t'$, we have $p_i\le \Delta^{-1/2}$, hence $\Lambda_i\le \sqrt{\Delta}$.
 Unrelated, due to $\Lambda_{i+1}=\sqrt{\Lambda_{i}}$ it suffices to prove  $\Lambda_{i}\le \sqrt{\Delta/\log n}\cdot \Lambda_i^{1/4}$, which reduces to $\Lambda_i\le (\Delta/\log n)^{2/3}$. The latter holds for $\Lambda_i\le \sqrt{\Delta}$, since $\Delta>\log^4 n$. \qedhere
\end{enumerate}
\end{proof}

\paragraph{Synchronized Color Trial.} In {\synchronizedcolortrial} (see Alg.~\ref{alg:synchtrial}), each vertex of $R_i^C$ is assigned a random candidate color distinct from those of other nodes in $R_i^C$, which it then tries via {\trycolor}. To achieve this, each vertex of $R_i^C$ selects a random subset $P(v)$ of its palette $\Psi(v)$  and sends $P(v)$ to the clique leader $w_C$, who then locally processes the nodes in $R_i^C$ in an arbitrary order and assigns each node $v$ a random color from its set $P(v)$ that is different from the colors assigned to previous nodes. We let $|P(v)|=\Pi_i=c_P\max(1,|R_i^C|/\Lambda_{i+1})\log n$, where $c_P>0$ is a large enough constant, specified in \Cref{lem:colorpickprobability}.
\begin{algorithm}[H]\caption{{\synchronizedcolortrial} (executed in  $R_i^C$, for all $C$ in parallel)}
\label{alg:synchtrial}
  \begin{algorithmic}[1]
  \STATE Leader $w_C$ computes  $|R_i^C|$ and sends it to all nodes in $R_i^C$.
     \STATE Each node $v\in R^C_i$ sends a uniformly random subset $P(v)\subseteq\Psi(v)$  of size $|P(v)|= \Pi_i$ to $w_{C}$. 
    \STATE $w_C$ processes the nodes in $R_i^C$ in an arbitrary order $v_1,v_2\dots$, where $v_j$ is assigned a candidate color $c_j$ chosen uniformly at random from $ P(v_j) \setminus \{c_1, c_2, \ldots, c_{j-1}\}$. 
    \STATE $w_C$ sends each node $v_j$ its  candidate color $c_j$.
    \STATE {\trycolor}($v_j$, $c_j$) in $G[R_i]$, for all $j\ge 1$. \label{st:trying}

  \end{algorithmic}
\end{algorithm}

We show in  \Cref{ssec:randomGreedyImplementation}  that one iteration of {\synchronizedcolortrial} can be implemented in the {\CONGEST} model in $O(1)$ rounds, using the \emph{clique overlay} computed in step~\ref{st:overlay} of Alg.~\ref{alg:main}.
Note that this is nearly trivial to do in the \LOCAL model, as any node in $R_i^C$ is within distance $2$ from the clique leader $w_C$. Further, each node can send its whole palette $\Psi(v)$ to the leader, which simplifies the analysis.

\subsection{Degree Reduction for Dense Nodes}
\label{ssec:degreeReduction}
We now prove that $O(\log\log \Delta)$ repetitions of {\synchronizedcolortrial} (line~\ref{st:core})  
reduces the external degree of each node in $R_i$, and 2 more iterations reduce the number of neighbors of each node $u\in V$ in $R_i$ to $O(\log^2 n)$. 

We need the following notation. 
For $v\in V$, let $r_i(v)=N(v)\cap R_i$ denote its number of uncolored neighbors in $R_i$. For $v\in R_i^C$, let $e_i(v)$ denote its \emph{external degree} -- its number of uncolored neighbors in $R_i\setminus R_i^C$ -- and $a_i(v)$ its \emph{antidegree} -- its number of non-adjacent nodes in $R_i^C$. 
Note that $e_i(v), r_i(v)$ and $a_i(v)$ may change during the execution of the algorithm, but only downwards.

We begin by bounding the size of $R^C_i$, as well as various degrees of nodes, restricted to $R_i^C$, immediately after the partitioning. 

\begin{lemma}\label{lem:nbrsinri}
Let $C$ be an almost-clique, and $0\le i\le t$. After line~\ref{st:partition} of Alg.~\ref{alg:dense}, it holds w.h.p.\ that: (i) $|R_i^C|=\Theta(\Lambda_i)$, (ii) $r_i(u)=O(\Lambda_i)$, for $u\in V$, (iii)  $|N(v)\cap R_i^C|=\Theta(\Lambda_i)$  and  $e_i(v)=O(\Lambda_i)$, for  $v\in C$.
\end{lemma}
\begin{proof}
Initially, $|C|=(1\pm \eps)\Delta$, and each node $v\in C$ has at least $(1-\eps)\Delta$ neighbors in $C$. Before partitioning, a node in $C$ can only be colored in {\slackgeneration}, where it participates with probability $p=\frac{1}{20}$. Thus, every node $v\in C$ joins $R^C_i$ w.p.\ $p'_i\ge \frac{19p_i}{20}$, and similarly, every node $v\in V\setminus V_{sparse}$ joins $R_i$ independently w.p.\ $p'_i$. The expected size of $R_i^C$, as well as the expected number of neighbors of a vertex $v\in C$ in $R^C_i$ is $\Theta(p_i\Delta)=\Theta(\Lambda_i)$. Similarly, the expected number of neighbors of a node $u\in V$ in $R_i$ is $O(p_i\Delta)=O(\Lambda_i)$. Since $\Lambda_i>c\log n$, for a large enough constant $c$ (Obs.~\ref{obs:lambdaprops}), all claims follow by using Chernoff bound (\ref{eq:chernoffless}). (Note that $e_i(v)\le r_i(v)$.)
\end{proof}

The following lemma highlights how many ``free'' colors each node in $R_i^C$ has, both due to its slack, as well as due to the fact that during its processing, its neighbors in $R_{i+1}^C$ stay uncolored.

\begin{lemma}[candidate color assignment]\label{lem:colorpickprobability}
Let $0\le i<t$. If $\Pi_i=c_P\max(1,|R_i^C|/\Lambda_{i+1})\log n$, for a large enough constant $c_P>0$, then  in every iteration of {\synchronizedcolortrial} in layer $i$, each node $v\in R_i^C$ has a palette of size at least $|R_i^C|+3e_i(v)+\Omega(\Lambda_{i+1})$ and is assigned a candidate color from $P(v)$, w.h.p., even if conditioned on an arbitrary color assignment to other nodes in $R_i^C$. 
\end{lemma}
\begin{proof}
Let $v\in C$ be a node. Let $a_i(v)=a_0$ and $e_i(v)=e_0$  be the initial 
antidegree and external degree of $v$ before slack generation (line~\ref{st:mainslack} in Alg.~\ref{alg:main}). By Lemma~\ref{lem:sparseGetsSlack}, after  slack generation, $v$ has slack $c(a_0+e_0)$,  for a constant $c>0$, w.h.p. By Lemma~\ref{lem:nbrsinri}, $v$ also has $c_1\Lambda_{i+1}$ uncolored neighbors in $R_{i+1}^C$, for a constant $c_1>0$, w.h.p.; hence, $v$ initially has a palette of size $|\Psi(v)|\ge |N(v)\cap R_i^C|  + c(a_0+e_0) + c_1\Lambda_{i+1}$. If $a_0+e_0=\Omega(\log n)$, with a large enough coefficient, then, by Lemma~\ref{lem:oneshotslowdecrease}, after  $O(1)$ executions of {\rct} (line~\ref{st:oneShotFor2} in Alg.~\ref{alg:dense}), we have $a_i(v)<ca_0$ and $e_i(v)<(c/3)e_0$, w.h.p., and  
the palette size in each subsequent round is 
\begin{align}
|\Psi(v)|&\ge |N(v)\cap R_i^C| + c(a_0+e_0) + c_1\Lambda_{i+1}\ge |N(v)\cap R_i^C|+a_i(v)+3e_i(v)+c_1\Lambda_{i+1}\notag\\
&\ge |R_i^C| + 3e_i(v) + c_1\Lambda_{i+1}\ .\label{eq:palettesizeslack}
\end{align}
If, on the other hand, $a_0+c_0=O(\log n)$, then (\ref{eq:palettesizeslack}) still holds, with a different constant $c'_1=c_1/4$ in front of $\Lambda_{i+1}$, since by Obs.~\ref{obs:lambdaprops}, $\Lambda_{i+1}=\Omega(\log n)$, with a large enough coefficient, which we can choose so that $c_1\Lambda_{i+1}>4(a_0+e_0)>c_1\Lambda_{i+1}/4 + 3(a_0+e_0)$.
 Let $\mathcal{C}$ be the set of colors that the leader assigns to the nodes in $R_i^C$ preceding $v$. It follows from (\ref{eq:palettesizeslack}) that for any set  $\mathcal{C}$, the probability that $P(v)\subseteq \mathcal{C}$ is at most  
\[
\left(\frac{|R_i^C|}{|R_i^C| + c_1\Lambda_{i+1}}\right)^{|P(v)|}=\left(1-\frac{c_1\Lambda_{i+1}}{|R_i^C|+c_1\Lambda_{i+1}}\right)^{|P(v)|}\le \exp\left(-h\right)\ ,
\]
where $h=c_1|P(v)|\Lambda_{i+1}/(c_1\Lambda_{i+1}+|R_i^C|)>c_3\log n$, for a large enough constant $c_3>0$, since $|P(v)|=\Pi_i=c_P\max(1,|R_i^C|/\Lambda_{i+1})\log n$, for a large enough constant $c_P>0$. Thus, 
it holds w.h.p.\ that the leader assigns $v$ a color from $P(v)\setminus \mathcal{C}$.  
\end{proof}

\begin{lemma}\label{lem:exdegreeSmall}
Let $i<t$ and assume $R_0,\dots,R_{i-1}$ have been colored. After $O(\log\log \Delta)$ iterations of {\synchronizedcolortrial}, it holds for every node $v\in R_i^C$ that $e_i(v)=O(\log n)$, w.h.p.
\end{lemma}
\begin{proof} 
The goal is to apply Lemma~\ref{lem:basiconeshot} with a suitable subgraph $H$. To obtain $H$ from $G[R_i]$ each node $v\in R_i^C$ removes all but the edges to nodes in $R_i\setminus R_i^C$. Note that the uncolored degree of a vertex in $H$ corresponds to its (uncolored) external degree in $G[R_i]$.
When restricted to $H$ the algorithm (line~\ref{st:core} in Alg.~\ref{alg:dense}) is still a valid coloring algorithm, since nodes in $V_{sparse}\cup \cup_{j\neq i} R_j$ do not participate, and pairs of nodes within the same almost-clique $C$ are never assigned the same color.   
To apply Lemma~\ref{lem:basiconeshot} we need to show that in iteration $j$ of {\synchronizedcolortrial} each vertex $v$ gets colored with probability at least $1-N_j(v)/s_v$ irrespective of the color choices of other nodes, for a suitable choice of $s_v$, where $N_j(v)$ are the uncolored neighbors of $v$ in iteration $j$ in $H$.  
 We condition the rest of the proof on the high probability event that every node in $R_i$ in every iteration is assigned a candidate color (due to \Cref{lem:colorpickprobability}).

Let $v\in R_i^C$ be a vertex, and consider a fixed iteration of {\synchronizedcolortrial}. Let $\mathcal{C}$ be the (random) set of candidate colors assigned to the nodes of $R_i^C$ preceding $v$. In the rest of the paragraph, we condition on an arbitrary outcome of $\mathcal{C}$. By the assumption above, 
node $v$ is assigned a candidate color $c_v\in P(v)\setminus \mathcal{C}$. 
  It follows from the randomness of $P(v)$ and the random choice of the candidate color $c_v$, that $c_v$ is uniformly distributed in $\Psi(v)\setminus \mathcal{C}$. Let $\mathcal{C'}$ be the set of colors assigned to the external neighbors of $v$. It follows that even when conditioned on arbitrary $\mathcal{C'}$, $v$ is permanently colored, i.e., $c_v\in \Psi(v)\setminus (\mathcal{C}\cup \mathcal{C'})$ holds, w.p.\ at least $1-|\mathcal{C'}|/|\Psi(v)\setminus \mathcal{C}|$. We know that $|\mathcal{C}|<|R_i^C|$, $|\mathcal{C'}|\leq|N_H(v)|$, and  $|\Psi(v)|\ge |R_i^C|+3e_i(v)+\Omega(\Lambda_{i+1})$ (from \Cref{lem:colorpickprobability}). Thus, $v$ gets permanently colored w.p.\ at least \begin{equation}\label{eq:cliquesuccprob}
  1-\frac{|\mathcal{C'}|}{|\Psi(v)\setminus \mathcal{C}|}\geq 1-\frac{|N_j(v)|}{2e_i(v)+\Omega(\Lambda_{i+1})}=1-\frac{|N_j(v)|}{s_v}\ , 
  \end{equation} 
irrespective of the candidate color assignments of other nodes in $G$ and where $s_v=2e_i(v)+\Omega(\Lambda_{i+1})$.

Since also $d_v=|N_j(v)|= e_i(v)\le s_v/2$ and $s_v=\Omega(\Lambda_{i+1})=\Omega(\log n)$, with a large enough constant factor (provided by Obs.~\ref{obs:lambdaprops}), 
Lemma~\ref{lem:basiconeshot} applies with parameters $d_v$ and $s_v$. Note that $s_v\in O(\Lambda_i)\cap \Omega(\Lambda_{i+1})$, so after $O(\log\log \min s_v)=O(\log\log \Delta)$ iterations, each node in $R_i$ has $e_i(v)=O(\max_{u,w\in R_i} (s_u/s_w)\cdot  \log n)=O(\Lambda_{i+1}\log n)$, w.h.p. 

Now, we bound $e_i(v)$ further by  applying  \Cref{lem:basiconeshot} again with a smaller upper bound on $\max_{u,w\in R_i} (s_u/s_w)$: Let us replace $s_v=\min(s_v, O(\Lambda_{i+1}\log n))$; note that we still have $s_v\ge 2d_v=2e_i(v)$, 
as well as $s_v=\Omega(\Lambda_{i+1})$, and hence, $\max_{u,w\in R_i}(s_u/s_w)=O(\log n)$. Applying \Cref{lem:basiconeshot} again with the same reasoning and parameters $d_v=e_i(v)$ and $s_v$, we see that after  $O(\log\log\Delta)$ more iterations,   $e_i(v)=O(\log^2 n)$,  w.h.p. A final application of \Cref{lem:basiconeshot} with $d_v=e_i(v)$ and $s_v=O(\log^2 n)$ implies $e_i(v)=O(\log n)$, in $O(\log\log\log n)$ more iterations, w.h.p.
\end{proof}

\begin{lemma}\label{lem:degreeSmall}
Let $i<t$ and assume $R_0,\dots,R_{i-1}$ have been colored, and for every node $v\in R_i$, $e_i(v)=O(\log n)$. After 2 iterations of {\synchronizedcolortrial}, it holds, for every node $u\in V$, that $r_i(u)=O(\log^2 n)$, w.h.p.
\end{lemma}
\begin{proof}
We condition on the high probability events that for each node $u\in V$, $r_i(u)=O(\Lambda_i)$ (\Cref{lem:nbrsinri}), each $v\in R_i^C$ has palette size $|\Psi(v)|\ge |R_i^C| + 3e_i(v)+\Omega(\Lambda_{i+1})$  and receives a candidate color in every iteration (\Cref{lem:colorpickprobability}). 
Consider a node $u\in V$ and an arbitrary iteration of {\synchronizedcolortrial}. For every neighbor $v\in N(u)\cap R_i$, let $X_v$ be a binary random variable that is 1 iff $v$ stays uncolored in the iteration. 
As it was shown in (\ref{eq:cliquesuccprob}),  $Pr[X_v=1]\le e_i(v)/(2e_i(v)+\Omega(\Lambda_{i+1}))=O(\log n/\Lambda_{i+1})$, 
even when conditioned on an adversarial choice of candidate colors for nodes in $R_i\setminus \{v\}$, and hence, when conditioned on arbitrary values of $X_u$, $u\in R_i\setminus \{v\}$ (since \emph{trying a color} is deterministic). 
Recalling that $r_i(u)=O(\Lambda_i)$, we have  $E\left[\sum_{v\in N(u)\cap R_i} X_v\right]=\frac{O(\Lambda_i\log n)}{\Lambda_{i+1}}$, and since the latter is $\Omega(\log n)$, an application of Chernoff bound (\ref{eq:chernoffmore}) implies that after the iteration, $r_i(u)=|N(u)\cap R_i|=\sum_{v\in N(u)\cap R_i} X_v=\frac{O(\Lambda_i\log n)}{\Lambda_{i+1}}$, w.h.p.
By the same reasoning ($s_v$ does not decrease), after another iteration, $r_i(u)=O(\max(1, \Lambda_i\log n/\Lambda^2_{i+1})\log n)$, w.h.p., which  is in $O(\log^2 n)$, by Obs.~\ref{obs:lambdaprops}, (v).
\end{proof}

\subsection{CONGEST Implementation and Proof of Lemma~\ref{lem:coloringDense}}
\label{ssec:randomGreedyImplementation}
All steps in Alg.~\ref{alg:dense}, except for the candidate color assignment, can be implemented in {\CONGEST} by design. The computation and distribution of $|R_i^C|$ in {\synchronizedcolortrial} can be done in $O(1)$ rounds using standard aggregation tools and the fact that $C$ has diameter $2$ (Lemma~\ref{lem:acdproperties}). 
It remains to show that nodes can indeed send their sets $P(v)$ to their leader in $O(1)$ {\CONGEST} rounds, using the clique overlays of almost-cliques.  
Recall that colored nodes automatically leave $R_i, R_i^C$, reducing their size.
\begin{lemma}\label{lem:rcisize} 
After line \ref{st:smalldeg1} of Alg.~\ref{alg:dense}, $|R_0^C|<\Delta/\log^2 n$, w.h.p.
\end{lemma}
\begin{proof}
Initially, $|R_0^C|=O(\Delta)$. By Lemma~\ref{lem:oneshotslowdecrease}, after each application of {\rct} in Step \ref{st:lglgnoneshot}, the number of nodes in $R_0^C$ of degree $\Omega(\log n)$ decreases by a constant factor, w.h.p. Hence, after $O(\log \log n)$ iterations, it decreases by a factor $O(\log^2 n)$, w.h.p. Since we also color (and remove) the small degree vertices in step~\ref{st:smalldeg1}, the size of $R_0^C$ shrinks to $\Delta/\log^2 n$, w.h.p.
\end{proof}

\begin{lemma}\label{lem:runtime}
W.h.p., in each iteration of {\synchronizedcolortrial} in $R_i^C$, $i<t$, each node succeeds to send a sub-palette $P(v)$ of size  $\Pi_i$ to  leader $w_C$ in $O(1)$ rounds. 
\end{lemma}
\begin{proof}
We use the clique overlay of almost-clique $C$, provided by \Cref{thm:congestedclique}, together with Lenzen's scheme~\cite{Lenzen13}. 
To this end, we need to ensure that each node has to send and receive $O(\Delta)$ colors, and that it indeed has $\Pi_i= c_P\max(1,|R_i^C|/\Lambda_{i+1})\log n$ colors in its palette. The latter follows from Lemma~\ref{lem:colorpickprobability}: each node $v\in R_i^C$ has palette size at least 
$|R_i^C|+\Omega(\Lambda_{i+1})=\Omega(\max(1,|R_i^C|/\Lambda_{i+1})\log n)$, w.h.p., where the equality holds as long as $\Lambda_{i+1}>c\log n$, for a large enough constant $c$ (provided by Obs.~\ref{obs:lambdaprops}).

Consider a fixed iteration of {\synchronizedcolortrial}. Note that $\Pi_i=O(\Delta)$, since $|R_i^C|<(1+\eps)\Delta$ and $\Lambda_{i+1}=\Omega(\log n)$, as observed above; hence, it remains to show that $w_C$ has to receive only $O(\Delta)$ colors. 
There are $|R_i^C|$ uncolored nodes in $C$, each sending $c_P \max(1,|R_i^C|/\Lambda_{i+1})\log n$ colors to $w_C$. If $|R_i^C|<\Lambda_{i+1}$, then  $w_C$ has to receive $O(|R_i^C|\log n)=O(\Delta)$ colors, since by Lemma~\ref{lem:rcisize},  $|R_0^C|<\Delta/\log^2 n$, and by Lemma~\ref{lem:nbrsinri}, $|R_i^C|=O(\Lambda_i)=O(\Delta/\log n)$, for $i>0$, w.h.p. Otherwise, 
$w_C$ has to receive $O(|R_i^C|^2\log n/\Lambda_{i+1})$ colors, which again is in $O(\Delta)$: for $i>0$, we have $\Lambda_i^2\log n/\Lambda_{i+1}\le \Delta$,  by Obs.~\ref{obs:lambdaprops},~\ref{obsi:sending}, while for  $i=0$, we have, by Lemma~\ref{lem:rcisize},  $|R_0^C|<\Delta/\log^2 n$, w.h.p., and $\Lambda_1=\Delta/\log^{3/2} n$, so $|R_0^C|^2\log n/\Lambda_1<\Delta$.
This completes the proof. 
\end{proof}

\begin{proof}[Proof of Lemma~\ref{lem:coloringDense}]

\Cref{lem:exdegreeSmall,lem:degreeSmall} imply that w.h.p., after processing layer $i<t$, each node $u\in V$ has at most $O(\log^2 n)$ neighbors in $R_i$. Similarly, \Cref{lem:nbrsinri} implies that every node $u\in V$ has $O(\Lambda_t)=O(\log^2 n)$ neighbors in $R_t$, w.h.p. Altogether, after processing all layers, every node $u\in V$ has at most $O(t\log^2 n)=O(\log\log\Delta \cdot \log^2 n)$ neighbors in $V\setminus V_{sparse}$; hence, {\colorsmalldegreenodes} (line~\ref{st:smalldeg3}) successfully colors $G[V\setminus V_{sparse}]$, w.h.p.

Partitioning into layers can be done without communication.  
The preprocessing in line~\ref{st:lglgnoneshot} takes $O(\log\log n)$ rounds, followed by $T$ rounds for {\colorsmalldegreenodes}, which is applied on a graph with $n$ nodes, maximum degree $\poly \log(n)$ and color space size $\poly(n)$.
We iterate through $O(\log\log \Delta)$ layers, and each layer makes $O(1)$ iterations of {\rct} (each taking $O(1)$ rounds), followed by $O(\log\log\Delta)$ iterations of {\synchronizedcolortrial} (each taking $O(1)$ rounds).
Finally, we apply {\colorsmalldegreenodes} to a graph with $n$ nodes, maximum degree $\poly\log(n)$ and color space size $\poly(n)$.
Thus, the runtime is in 
$O\big(\log\log n + (\log\log \Delta)^2 +T\big)$.
\end{proof}

\section{Computation of ACD and Clique Overlay }\label{sec:ACD}

\subsection{ACD Computation in Congest}
\label{ssec:acdComputation}
A distributed algorithm \emph{computes an ACD} if each $C_i$ has a unique AC-ID that is known to all nodes in $C_i$, and each node in $V_{sparse}$ knows it is in $V_{sparse}$.  The goal of the remainder of this section is to prove the following lemma.
\begin{lemma}[ACD computation] 
\label{lem:ACD}
There is an $O(1)$ round randomized CONGEST algorithm that, given a graph $G=(V,E)$ with maximum degree $\Delta\ge c\log^{2}n$, for a sufficiently large constant $c$, computes an $(\eps,\eta)$-ACD with $\eps=1/3$ and $\eta=\eps/108$, w.h.p. 
\end{lemma}
The algorithm {\computeacd}, described in Alg.~\ref{alg:acd}, computes an ACD, using a bootstrapping-style procedure. We begin by sampling a subset $S$ of vertices, where each vertex is independently sampled w.p.\ $1/\sqrt{\Delta}$ (line~\ref{st:acdsample}). The idea is to use the whole graph to relay information between the $S$-nodes, in order to efficiently find out which pairs of $S$-nodes belong to the same almost-clique (to be constructed), then let each node  $v\in V\setminus S$ decide whether it joins such an $S$-almost-clique $C$, depending on the size of $N(v)\cap C$.  To construct the $S$-almost-cliques, each node $v\in V$ chooses a random $S$-neighbor (ID) and broadcasts it to its neighbors (line~\ref{st:acdgossip}). A key observation is that if two $S$-nodes are similar, i.e., share many neighbors, then they will likely receive each other's IDs many times, which allows them to detect similarity (with some error), as well as friend edges (lines~\ref{st:acdsimilar}--\ref{st:acdfriend}). The $S$-nodes that have detected many friend edges are likely to be dense, and join a set $S_{dense}$ (line~\ref{st:acddense}). Roughly speaking, the connected components induced by $S_{dense}$ and detected friend edges form the $S$-almost-cliques.  Each node in $S_{dense}$ selects the minimum ID of an $S_{dense}$-neighbor and sends it to the other neighbors, thus proposing to form an almost-clique with that ID (line~\ref{st:acdsendmin}). Each node in $V$ that receives the same ID from many neighbors also joins the almost-clique with that ID (line~\ref{st:formclique}).  
 We prove that each connected component $C$ that contains a sufficiently dense node has the desired properties (w.h.p.): $|C|\approx \Delta$, and each node $v\in C$ has $\approx \Delta$ neighbors in $C$. The components not satisfying the required properties are moved into $V_{sparse}$ (line~\ref{st:acdfilter}), as they do not contain a dense node (w.h.p.).

Throughout this section we let $\delta=\eps/27<1/80$, and require $\Delta\geq c\log^{2}n$, for a constant $c>0$, as in Lemma~\ref{lem:ACD}. Our claims hold for smaller $\delta$, provided that $\Delta>(c/\delta^4)\log^2 n$, for a large enough constant $c>0$.
For any $\gamma\in [0,1]$, we let $F_{\gamma}$ be the set of $\gamma$-friend edges in $G$, and  let $G_{F_{\gamma}}=(V,F_{\gamma})$ be the graph induced by $F_\gamma\cup V$.

\begin{algorithm}[H]\caption{{\computeacd} ($\delta\in (0,1/80)$)}\label{alg:acd}
\begin{algorithmic}[1]
\STATE \textbf{Sample} Each node selects itself into $S$ independently w.p.\ $1/\sqrt{\Delta}$ and notifies its neighbors.\label{st:acdsample}
\STATE \textbf{Gossip} Each node $v\in V$ selects uniformly at random one ID of an $S$-neighbor (if it has any) and forwards it with probability $\min(1, |S\cap N(v)|/(2\sqrt{\Delta}))$ to each of its $S$-neighbors.\label{st:acdgossip}
\STATE \textbf{Detect similar nodes} Each $S$-node $v$ decides it is \emph{approximately $\delta$-similar} to node $u$ if $v$ receives the ID of $u$ at least $(1-2\delta)\sqrt{\Delta}/2$ times.\label{st:acdsimilar}
\STATE \textbf{Detect friendship edges} The set $F\subseteq E$ consists of the edges between vertices in $S$ where at least one endpoint identified the other one as approximately $\delta$-similar.\label{st:acdfriend}

\STATE \textbf{Detect density} Each $S$-node with more than  $(1-2\delta)\sqrt{\Delta}$ incident $F$-edges joins a set $S_{dense}$.\label{st:acddense}

\STATE{Each node  $v\in S_{dense}$  selects the minimum ID of a node $u\in S_{dense}$ such that $\{u,v\}\in F$ and broadcasts it to its neighbors in $G$.}\label{st:acdsendmin}

\STATE{\textbf{Form  cliques} Any vertex $v\in V$ (including vertices in $S$) that receives the same ID at least  $(1-11\delta)\sqrt{\Delta}$ times adapts this ID as its AC-ID, otherwise $v$  joins $V_{sparse}$.\label{st:formclique}}
\STATE{\textbf{Remove bad cliques} Use the AC-ID node $w$ of each connected component $C$ as a leader node to ensure: If $|C|<(1-\delta)\Delta$ or there is a node $u\in C$ with $|N(u)\cap C|<(1-27\delta)\Delta$ then all nodes in $C$ join $V_{sparse}$.}\label{st:acdfilter}
\end{algorithmic}
\end{algorithm}

In the following analysis, we let $H=(S_{dense}, F\cap S_{dense}\times S_{dense})$ denote the subgraph induced by dense nodes ($S_{dense}$) and friendship edges ($F$).
We first show that vertices in $S$ detect friend edges between them (with some error). 

\begin{lemma}[$S$-friend edges are detected] 
\label{lem:friendsDetected}
Let $\delta< 1/4$ and $\Delta>(c/\delta^4)\log^2 n$, for a large enough constant $c>0$.
W.h.p., (1) any $\delta$-friend edge between vertices in $S$ is contained in $F$, and (2) any edge in $F$ is a $4\delta$-friend edge. 
\end{lemma}
\begin{proof}
Let $u,v\in S$, and let $R_{u,v}$ be the set of common neighbors of $u$ and $v$. The rest of the proof is conditioned on the event that every node in $G$ has at most $(1+\delta)\sqrt{\Delta}$ neighbors in $S$. By Lemma~\ref{lem:nbrhdpreserved}, this holds w.h.p. (using the bound on $\Delta$). Given this, every node in $R_{u,v}$ independently sends the ID of $u$ to $v$ w.p.\ $1/2\sqrt{\Delta}$. Thus, the expected number of times $v$ receives the ID of $u$ is  $|R_{u,v}|/2\sqrt{\Delta}$.

\textbf{(1):}
Assume that $\{u,v\}$ is a $\delta$-friend edge.  It holds that $|R_{u,v}|\ge (1-\delta)\Delta$,  and hence $v$ receives the ID of $u$ at least $(1-\delta)\sqrt{\Delta}/2$ times in expectation, and at least $(1-\delta)^2\sqrt{\Delta}/2\ge (1-2\delta)\sqrt{\Delta}/2$ times, w.h.p., by a Chernoff bound (\ref{eq:chernoffless}). Claim (1)  follows by a union bound over all $\delta$-dense edges.

\textbf{(2):}
 Now, assume that $\{u,v\}$ is not a $4\delta$-friend edge. Then $|R_{u,v}|< (1-4\delta)\Delta$, and by a similar application of (\ref{eq:chernoffmore}), $v$ receives the ID of $u$ less than $(1+\delta)|R_{u,v}|/2\sqrt{\Delta}< (1+\delta)(1-4\delta)\sqrt{\Delta}/2\le (1-2\delta)\sqrt{\Delta}/2$ times, w.h.p. Similarly, $u$ receives the ID of $v$ less than $(1-2\delta)\sqrt{\Delta}/2$ times, w.h.p. Claim (2) now follows by a union bound over all edges that are not $4\delta$-dense. 
\end{proof}

In line~\ref{st:acddense} of Alg.~\ref{alg:acd} $S$ nodes use the number of incident friend edges to detect whether they are dense. The next lemma provides the guarantees of this step.
\begin{lemma}[Density detection in $S$]\label{lem:Sdense}
Let $\delta<4$ and $\Delta>(c/\delta^4)\log^2 n$, for a large enough constant $c>0$. W.h.p., (1) if $v\in S$ is  $\delta/2$-dense  then it is contained in $S_{dense}$, and (2) if $v\in S$ is  not $4\delta$-dense  then it is not contained in $S_{dense}$. 
\end{lemma}
\begin{proof}
\textbf{(1):} Let $v$ be a $\delta/2$-dense node and $R_v\subseteq N(v)$ be the set of its $\delta/2$-friends. We know that $|R_v|\ge (1-\delta/2)\Delta$. By Lemma~\ref{lem:friendsDetected}, all edges $\{u,v\}$ with $u\in R_v\cap S$ are in $F$, w.h.p. By Lemma~\ref{lem:nbrhdpreserved} (using the bound on $\Delta$), $|R_v\cap S|\ge (1-\delta)|R_v|/\sqrt{\Delta}\ge (1-\delta)(1-\delta/2)\sqrt{\Delta}\ge (1-2\delta)\sqrt{\Delta}$, and hence $v$ is in $S_{dense}$, w.h.p. Claim (1) now follows by a union bound over all $\delta/2$-dense nodes.

\textbf{(2):} Let $v$ be a node that is not $4\delta$-dense. Let $R_v\subseteq N(v)$ be the set of $4\delta$-friends of $v$. It holds that $|R_v|< (1-4\delta)\Delta$. By Lemma~\ref{lem:nbrhdpreserved} (using the bound on $\Delta$), $|R_v\cap S|<(1+\delta)(1-4\delta)\Delta\le (1-3\delta)\sqrt{\Delta}$, w.h.p.
By Lemma~\ref{lem:friendsDetected}, for $u\in (N(v)\setminus R_v)\cap S$ the edge $\{u,v\}$ is not in $F$, w.h.p. Hence, $v$ has at most $|R_v\cap S|< (1-3\delta)\sqrt{\Delta}$ incident edges in $F$, and hence $v$ does not join $S_{dense}$. Claim (2) now follows by a union bound over all nodes that are not $4\delta$-dense. 
\end{proof}

We will use the following observation from~\cite{ACK19} in the proof of Lemma~\ref{lem:delta4}.

\begin{observation}[\cite{ACK19}]
\label{obs:manyFriends} Let $\gamma<1/2$,
 $v$ be an $\gamma$-dense vertex,  and $S_v\subseteq N(v)$ be the $(1-\gamma)\Delta$ neighbors of $v$ along $\gamma$-friend edges. Every vertex in $S_v$ is $2\gamma$-dense.
\end{observation}
\begin{proof}
Let $w,w'\in S_v$. Since the edges $\{v,w\},\{v,w'\}$ are $\gamma$-friend edges, $|N(w)\cap N(v)|\ge (1-\gamma)\Delta$ and $|N(w')\cap N(v)|\ge (1-\gamma)\Delta$. By Obs.~\ref{obs:deltaintersections}, $|N(w)\cap N(w')|\ge (1-2\gamma)\Delta$. 
Similarly, since $|S_v|\ge (1-\gamma)\Delta$ and $|N(w)\cap N(v)|\ge (1-\gamma)\Delta$, we have $|N(w)\cap S_v|\ge (1-2\gamma)\Delta$. Thus, every node in $S_v$ has at least $(1-2\gamma)\Delta$ neighbors in $S_v$, all $2\gamma$-similar to it, i.e., it has at least $(1-2\gamma)\Delta$ of $2\gamma$-friends, as required. 
\end{proof}

In the next lemmas we show that the computed components satisfy the properties of an ACD.
\begin{lemma}\label{lem:delta4}
Let $\delta<1/20$ and $\Delta>(c/\delta^4)\log^2 n$, for a large enough constant $c>0$. Let $v$ be a $\delta/4$-dense node in $V$ and $R_v$ the set of its $\delta/4$-friends.  All vertices in $R_v\cup \{v\}$ output the same AC-ID, w.h.p.
\end{lemma}
\begin{proof}
 Fix $v'\in R_v$ and introduce the following sets:
\begin{align*}
  R_{v'} & =R_v\cap N(v')\\
S_{v,v'} & =R_{v'}\cap S\subseteq N(v)\cap N(v')\cap S	
\end{align*}
By Obs.~\ref{obs:manyFriends}, all nodes in $R_v$ are $\delta/2$-dense and as all nodes in $S_{v,v'}$ are contained in $R_v$, Lemma~\ref{lem:Sdense}, (1) implies that $S_{v,v'}\subseteq S_{dense}$. Thus, all nodes in $S_{v,v'}$ participate in the ID broadcasting in Step 6. 
Let $u$ be the node of minimum ID in $N_H(S_{v,v'})$, and let $X\subseteq S_{v,v'}$ be the set of nodes that have node $u$ as a neighbor in $H$. By the definition of $u$ (no node in $X$ has a smaller ID neighbor in $H$) all nodes in $X$ forward $u$'s ID to $v$ and $v'$ in Step 6. In the remainder of the proof, we show that $|X|\ge (1-11\delta)\sqrt{\Delta}$, w.h.p. The claim then follows by a union bound over all $v'$.

 Note that $|R_v|\ge (1-\delta/4)\Delta$ and $|N(v')\cap N(v)|\ge (1-\delta/4)\Delta$ (since $v,v'$ are $\delta/4$-friends), hence Obs.~\ref{obs:deltaintersections} implies that $|R_{v'}|\ge (1-\delta/2)\Delta$. 
  Let $w\in S_{v,v'}\subseteq R_v$ be a neighbor of $u$ in $H$, i.e., $\{u,w\}\in F$ ($w$ exists by the definition of $u$).
Since $w$ is a $\delta/4$-friend of $v$, $|N(w)\cap N(v)|\geq (1-\delta/4)\Delta$ holds, and since $|R_{v'}|\ge (1-\delta/2)\Delta$, Obs.~\ref{obs:deltaintersections} implies that $|N(w)\cap R_{v'}|\ge (1-\delta)\Delta$. 

Therefore, by Lemma~\ref{lem:nbrhdpreserved} (using the bound on $\Delta$), $|N(w)\cap S_{v,v'}|\ge (1-\delta)^2\sqrt{\Delta}\ge (1-2\delta)\sqrt{\Delta}$. Also by Lemma~\ref{lem:nbrhdpreserved}, $|N(w)\cap S|\le (1+\delta)\sqrt{\Delta}$. 
Edge $\{u,w\}$ is contained in $F$ (by the definition of $u$ and $w$) and by  Lemma~\ref{lem:friendsDetected}, (2) it is a  a $4\delta$-friend edge, w.h.p. Thus, $|N(w)\cap N(u)|\ge (1-4\delta)\Delta$. By Lemma~\ref{lem:nbrhdpreserved},  $|N(w)\cap N(u)\cap S| \ge (1-\delta)(1-4\delta)\sqrt{\Delta}\ge (1-5\delta)\sqrt{\Delta}$. 
Applying Obs.~\ref{obs:deltaintersections} to sets $C=N(w)\cap S$, $A=N(u)$ and $B=S_{v,v'}$, we see that $|N(w)\cap N(u)\cap S_{v,v'}|\ge (1-8\delta)\sqrt{\Delta}$.

From Lemma~\ref{lem:nbrhdpreserved}, $|N(u)\cap S|\le (1+\delta)\sqrt{\Delta}$. Let $T$ be the set of neighbors of $u$ along $F$-edges. Since $u$ is in $S_{dense}$, we have $|T|\ge (1-2\delta)\sqrt{\Delta}$. Applying Obs.~\ref{obs:deltaintersections} with $C=N(u)\cap S_{v,v'}$, $A=T$, and $B=S_{v,v'}$, we see that
$u$ is adjacent to at least $(1-11\delta)\sqrt{\Delta}$ nodes in $S_{v,v'}$, along edges belonging to $F$, i.e., $|X|\geq (1-11\delta)\sqrt{\Delta}$.
\end{proof}

\begin{lemma}\label{lem:commonw}
Let $\delta<1/13$ and $\Delta>(c/\delta^4)\log^2 n$, for a large enough constant $c>0$. Let $C$ be a component with AC-ID $w$, and $v\in C$. W.h.p.,  $|N(v)\cap N(w)|\ge (1-13\delta)\Delta$.
\end{lemma}
\begin{proof}
Let $v$ be a node such that there is a subset $S_v\subseteq N(v)\cap S_{dense}$ of size at least $(1-11\delta)\sqrt{\Delta}$ that broadcast the same ID $w$ in Step 6. 
Thus, $w$ and $v$ have $(1-11\delta)\sqrt{\Delta}$ common neighbors in $S$, and hence $|N(v)\cap N(w)|\ge (1-13\delta)\Delta$, as otherwise by Lemma~\ref{lem:nbrhdpreserved} (using the bound on $\Delta$), $|N(u)\cap N(w)\cap S|<(1+\delta)|N(v)\cap N(w)|/\sqrt{\Delta}<(1-11\delta)\sqrt{\Delta}$ (using $\delta<1/13$), w.h.p.
\end{proof}

\begin{lemma}[ACD forming]\label{lem:ACDFormation}
Let $\delta<29$ and $\Delta>(c/\delta^4)\log^2 n$, for a large enough constant $c>0$. The following hold w.h.p.:
\begin{enumerate}[label=(\arabic*)]
\item  every $\delta/4$-dense node in $V$ adapts an AC-ID. For every component $C$ that contains a $\delta/4$-dense node, 
\item   $|C|\ge (1-\delta/4)\Delta$,
\item  every node $v\in C$ has at least $(1-27\delta)\Delta$ neighbors in $C$,
\item  $|C|\le (1+25\delta)\Delta$.
\end{enumerate}
\end{lemma}
\begin{proof}

\textbf{(1), (2):} Immediately follows from Lemma~\ref{lem:delta4}.

\textbf{(3):} Let $C$ be the component with AC-ID $w$, and $v\in C$ be an arbitrary node in it. By Lemma~\ref{lem:commonw}, $|N(v)\cap N(w)|\ge (1-13\delta)\Delta$ holds w.h.p. Let $u\in C$ be the $\delta/4$-dense node in $C$. We know from Lemma~\ref{lem:delta4} that the set $R_u$ of at least $(1-\delta/4)\Delta$ neighbors of $u$ is contained in $C$. Since $|N(u)\cap N(w)|\ge (1-13\delta)\Delta$ and $|N(v)\cap N(w)|\ge (1-13\delta)\Delta$, we have $|R_u\cap N(w)|\ge (1-14\delta)\Delta$, and hence $|N(v)\cap N(w)\cap R_u|\ge (1-27\delta)\Delta$. The latter implies the claim since $R_u\subseteq C$.

\textbf{(4):} Following the proof of (3), let $w$ be the AC-ID of a component $C$ whose size we want to bound. By the definition of the algorithm, every node in $C$ is within distance 2 from $w$. Let $A=N(w)$ and $B=(N(A)\setminus A)\cap C$. Note that $A\cup B=C$. 
By Lemma~\ref{lem:commonw}, for each node $v\in C$, $|N(v)\cap N(w)|>(1-13\delta)\Delta$, w.h.p. 
Thus, for each node $u\in A$, $|N(u)\cap A|\ge (1-13\delta)\Delta$ and hence $|N(u)\cap B|<13\delta\Delta$. On the other hand, for each node $v\in B$, $|N(v)\cap A|\ge (1-13\delta)\Delta$. The former implies that the number of edges between $A$ and $B$ is at most $13\delta|A|$, while the latter implies that it is at least $(1-13\delta)|B|$. Hence, $|B|<\frac{13\delta}{1-13\delta} |A|<25\delta|A|$ which implies the claim. 
\end{proof}

\begin{proof}[Proof of Lemma~\ref{lem:ACD}]

The output is well-defined: Every node $v$ adopts at most one AC-ID, since by Lemma~\ref{lem:nbrhdpreserved}, $|N(v)\cap S|\le (1+\delta)\sqrt{\Delta}$, and $(1-11\delta)\sqrt{\Delta}>(1+\delta)\sqrt{\Delta}/2$, for $\delta<20$.
The main properties follow immediately from the algorithm definition and Lemma~\ref{lem:ACDFormation}, recalling that $\delta=\frac{\eps}{27}$.

Steps 1 to 6 clearly take $O(1)$ rounds. Step 7 takes $O(1)$ rounds, since by the definition of the algorithm, every node in a component with AC-ID  $w$ is within distance 2 from $w$; hence, $w$ can aggregate the size of $C$ and the degrees of nodes within $C$ in $O(1)$ rounds.
\end{proof}

\begin{remark}
The constants in Lemma~\ref{lem:acdproperties} can be improved by deriving the properties claimed there directly from our algorithm. However, that would make the exposition more complicated.
\end{remark}

\subsection{Clique Overlay for Almost-Cliques}\label{sec:congestedcliqueoverlay}

Consider a $\Delta$-clique $K_{\Delta}$ with a set of \emph{routing requests}, where each node has $O(\Delta)$ messages to send to or receive from other nodes in $K_\Delta$ (a node can have several messages addressed to the same destination). 
In his celebrated work, Lenzen~\cite{Lenzen13} designed an algorithm that allows an execution of any given set of routing requests in $O(1)$ {\CONGEST} rounds. In order to achieve a similar result for almost-cliques, we simulate a clique over a given almost-clique. Recall that an almost-clique has diameter 2 (\Cref{lem:acdproperties}). A \emph{clique overlay} of an almost-clique $C$ is a collection $O$ of length-2 paths, containing a path between any pair of non-adjacent nodes in $C$. The \emph{congestion} of an overlay $O$ is the maximum, over the edges $e$ of $G$, number of occurrences of $e$ in $C$. A distributed algorithm computes an overlay $O$ of an almost clique $C$ if for every  non-edge $\{u,v\}$ in $C$,  nodes $u,v$ know the node $w$ that forwards messages from $v$ to $u$ and vice versa in $O$.

Given a clique overlay of congestion 2 for an almost clique $C$, we can combine it with Lenzen's algorithm to execute in $O(1)$ rounds any given set of routing requests where each node sends and receives $O(\Delta)$ messages.

We give a randomized algorithm, called {\computecliqueoverlay}, that, for a given almost clique $C$,  computes a clique overlay with congestion 2 in $O(\log\log n)$ rounds. The main idea is to model the construction of an overlay as a list coloring problem, where missing edges of $C$ are the vertices of a graph, while the nodes of $C$ are the colors: a vertex $uv$ picking color $w$ means $w$ is used in the overlay to forward messages from $u$ to $v$ and back. As we note below, the runtime can be improved by using slightly more involved algorithms.

\begin{theorem}\label{thm:congestedclique}
Assume that $\eps\le 1/15$ and $\eps\Delta>c\log^2 n$, for a large enough constant $c>0$. There is a $O(\log\log n)$-round \CONGEST algorithm that, for any almost-clique $C$ with a leader node $w_C$, computes a clique overlay $O$ of congestion 2.
\end{theorem}

\begin{proof}

We assume the nodes in $C$ know the IDs of all non-neighbors in $C$. This can be achieved via standard aggregation methods, e.g., by forming a breadth-first-search tree rooted at the leader $w_C$, which  aggregates the sizes of subtrees, and counts and enumerates the nodes from $1$ to $|C|$. As a result, each node knows its new ID, as well as $|C|$, and in one additional round of communication, also knows the IDs of its neighbors; hence, it also knows the IDs of its non-neighbors. 

We reduce the construction of a congestion-2 overlay to a coloring problem with high bandwidth.
%
We form a graph $H$ with a vertex $uv$ for each non-edge $\{u,v\}$ in $G[C]$. Vertices $uv$, $wt$ of $H$ are adjacent if  
$\{u,v\} \cap \{w,t\} \ne \emptyset$. We refer here to vertices of $G$ as \emph{nodes} to distinguish from \emph{vertices} of $H$. 
The colorspace of our coloring problem is the set $[|C|]$ of recomputed node IDs in $C$. The initial palette $\Psi(uv)$ of a vertex $uv$ is the set $N(u)\cap N(v) \cap C$ of common neighbors of $u$ and $v$ in $C$. 

Observe that solving the constructed list coloring instance for $H$ yields a congestion-2 clique overlay $O$ for $C$. Namely, each non-edge (vertex) $uv$ is assigned a 2-path via some node $w$ (color), and each edge $\{u,w\}$ of $C$ is only used in 2-paths of color $u$ or $w$, and so it appear at most twice in $O$. Further note that every vertex in $H$ has degree at most $2\eps\Delta$, while the palette size is at most $(1-3\eps)\Delta$ (using Obs.~\ref{obs:deltaintersections}). As a result, assuming that $\eps\le 1/15$, each vertex has slack three times larger than its degree, that is,  we have a \emph{sparse node coloring} instance, with a caveat: the computation needs to be done in $G[C]$.

We have each vertex $uw$ of $H$ \emph{handled} by the incident node $u$ with higher ID. The handler $u$ has an imperfect information about the palette: during the algorithm, it maintains an \emph{apparent palette} $\Psi'(uv)=\{w\in N(u) : uw\notin O\}$, which contains some unusable colors, as $u$ does not know which of its neighbors $w\in N(u)$ are adjacent to $v$ ($\Psi'(uv)\setminus\Psi(uv)$ consists of those neighbors that are not adjacent to $v$). We let the handler nodes simulate {\rct} in $H$, with the apparent palettes of vertices, but only retain colors that belong to the original palettes. The execution is done in a sequence of pairs of rounds. In the first round of each pair, node $u$ picks a uniformly random color $c_v$ from each palette $\Psi'(uv)$ for each vertex $uv$ it handles. If the same color is sampled for two or more vertices $uv,uv',\dots$ then they are not colored in this round. For each remaining $uv$, $u$ sends the ID of $v$ to $c_v$ (note that $u$ sends a single message to $c_v$). In the second round of the pair, consider a node $w$ that receives node IDs $v_1,v_2,\dots,v_t$ from neighbors $u_1,u_2,\dots,u_t$. We assume that for each $i$, $w$ is adjacent to both $u_i$ and $v_i$; otherwise $w$ is not a usable color for $u_iv_i$, and ignores that pair. If $t=1$, $w$ broadcasts $u_1,v_1$ to its neighbors, indicating that $u_1v_1$ is colored with $w$, otherwise $w$ broadcasts a message indicating that color $w$ was not taken in the given round. All handler nodes update their palettes accordingly.

It is easy to see that the above correctly simulates {\rct} in $H$ with apparent palettes. 
A vertex $uv$ is colored in a given round if its random color lands in $\Psi(v)$ and is different from the colors picked by its neighbors; hence, if $d_0,d$ denote the number of uncolored neighbors of $uv$ at the beginning of the algorithm and in the current round, respectively, then $uv$ is colored w.p.
\[
\frac{|\Psi'(uv)|-|\Psi'(uv)\setminus\Psi(uv)|-d}{|\Psi'(uv)|}\ge 1-\frac{d_0+d}{|\Psi'(uv)|}\ge 1/2\ ,
\]
where the last inequality follows from the observation that each vertex has slack $3d_0$ (and slack never decreases). Moreover, the probability bound clearly holds even when conditioned on adversarial candidate color choices of neighbors. Thus, also using that $\Delta=\Omega(\log^2 n)$, with a large enough coefficient, we can apply Chernoff bound~(\ref{eq:chernoffmore}) to show that while there are at least $\frac{\eps\Delta}{3\log n}$ uncolored vertices, in each iteration, the number of uncolored vertices  decreases by a constant factor, w.h.p. (we use the assumption that $\eps\Delta=\Omega(\log^2 n)$, with a large enough coefficient); therefore, in $O(\log\log n)$ rounds there are at most $\frac{\eps\Delta}{3\log n}$ vertices remaining to be colored, w.h.p. To finish coloring, in one more round, each node $u$ tries $3\log n$ colors in parallel, for each vertex $uv$ it handles. The probability of success of each trial is at least $1/2$, by the same analysis as above, since we can view parallel trials by the same vertex as trials by $3\log n$ different vertices; even so, each vertex will have degree bounded by $2\eps\Delta$, and slack at least three times that, as before. Thus, each parallel trial for a given vertex succeeds w.p. at least $1/2$, irrespective of the outcome of other trials (including the trials for other vertices). Since there are at least $3\log n$ trials per vertex, Chernoff bound~(\ref{eq:chernoffless}) implies that each vertex is successfully colored, w.h.p. 
\end{proof}
\begin{remark}
In the proof above, the ability of a node $u$ to locally resolve  conflicts between the candidate color picks by the vertices $uv$ it handles makes the coloring problem more like a \LOCAL coloring, since every vertex can now try many colors in parallel. This  suggests that the multi-trial technique of~\cite{SW10} can be applied to get a $O(\log^* n)$ time algorithm. A quicker improvement can be obtained by using a variant of \Cref{lem:basiconeshot}, which, given that all nodes have equal initial slack $\Omega(\Delta)$, reduces their degrees to $O(\Delta/\log n)$ in $O(\log\log\log n)$ rounds. Note that  we need to use parallel trials here too, in order to suppress the failure probability caused by the difference between the apparent and real palettes.
\end{remark}

\section{Coloring Small Degree Graphs}\label{sec:smalldegree}
In this section, we show how to $(deg+1)$-list color a graph in $O(\log\Delta)+\poly\log\log(n)$ rounds. This result is efficient in the small degree case, i.e., when the maximum degree is bounded by $\poly\log(n)$, with runtime reducing to $O(\log^5\log n)$ rounds. We call the obtained algorithm {\colorsmalldegreenodes}.
\begin{theorem}\label{thm:smallDegree}
Let $H$ be a subgraph of $G$ with maximum degree $\Delta_H$. Assume that each node $v$ of degree $d_v$ in $H$ has a palette $\Psi(v)\subseteq [U]$ of size $|\Psi(v)|\ge d_v+1$ from a colorspace of size $U=\poly(n)$. 
There is a $O(\log\Delta_H+\log^5\log n)$-round randomized algorithm that w.h.p. colors $H$ by assigning each node $v$ a color from its palette.
\end{theorem}

Similar results are known in the \LOCAL model  \cite{BEPSv3,RG19,GGR20} and the \CONGEST model \cite{Ghaffari2019,GGR20}, where the latter has slightly worse runtime. We provide a (mostly) self contained proof of the result, meantime fixing an error in the state-of-the-art \CONGEST algorithm (see Remark~\ref{rem:smallError}) and slightly improving the runtime, and thus matching the runtime of the state-of-the-art in the \LOCAL model.
 Our  algorithm consists of four steps: \emph{Shattering}, \emph{Network Decomposition}, \emph{Colorspace Reduction} and \emph{Cluster Coloring}. 
These four steps have been used in a similar (but not identical) manner in \cite{Ghaffari2019}.   Colorspace reduction is an adaptation of a result in \cite{HKMN20}, and the shattering part stems from \cite{BEPSv3}.

\paragraph{High Level Overview}
The \emph{shattering} part consists of $O(\log \Delta_H)$ iterations of {\rct}. This results in a partial coloring of the graph, after which the uncolored parts of $H$ form connected components of size  $N=\poly\log(n)$, w.h.p., as shown in \cite{BEPSv3}. In the remainder of the algorithm, such  components are processed in parallel and independently from each other. 

In the \emph{network decomposition } part, each component is partitioned into $\Gamma_1,\ldots,\Gamma_c$ collections of \emph{clusters}, with $c=O(\log\log n)$. The important properties of this decomposition are: (i) the clusters in each $\Gamma_i$ are independent, i.e., there is no edge between the nodes of any two distinct $Q, Q'\in\Gamma_i$, (ii) each cluster $Q\in \Gamma_i$ \emph{essentially} has a $\poly\log\log(n)$ diameter. By Property (i), all clusters in $\Gamma_i$ can be colored in parallel without interference, while Property (ii) is used to do that efficiently. 

The algorithm that we use to color the clusters has runtime depending on the size of the colorspace. In order to make the coloring procedure efficient, we perform \emph{color reduction}, where we deterministically compute a function $f_Q$ that maps the color palettes of all nodes in $Q$ to a much smaller colorspace of size $\poly(N)$, while preserving the palette sizes. In the new color space, each color can be represented with $O(\log\log n)$ bits.  This step is run on all clusters  in all $\Gamma_i$ in parallel.

In the \emph{cluster coloring part}, we iterate through the groups $\Gamma_1,\ldots,\Gamma_c$. When processing $\Gamma_i$ each cluster $Q\in \Gamma_i$ is colored in parallel. Note that we need to solve a $(deg+1)$-list coloring on $Q$: all colors of previously colored neighbors in a cluster in $\Gamma_1,\ldots,\Gamma_{i-1}$ are removed from the palettes of nodes in $Q$. In the last step, we solve the list coloring problem on $Q$ by simulating $O(\log n)$ parallel and independent instances of the simplest coloring algorithm: iterate {\rct}, for $O(\log N)$ rounds; note that the size of $Q$ is upper bounded by $N=\poly\log(n)$. The failure probability of each instance is $1/\poly(N)$, and   w.h.p.\ (in $n$), at least one of these instances is successful~\cite[Sec. 4]{Ghaffari2019}. After executing all instances for $O(\log N)$ iterations, the nodes in $Q$ agree on a successful instance and get permanently colored with the colors chosen in that instance. 

\smallskip

The color space reduction is necessary for sending the random color trials of several parallel instances in one $O(\log n)$-bit \CONGEST message.  Alternatively, one could replace the cluster coloring procedure with the deterministic $(deg+1)$-list coloring algorithm from \cite{BKM19}. However, its runtime depends logarithmically on the color space, and one would need the same color space reduction to obtain a $\poly\log\log(n)$ round algorithm. The runtime of our approach is dominated by the time to compute a network decomposition.

\paragraph{Shattering}  Barenboim, Elkin, Pettie, and Su \cite{BEPSv3} showed that $O(\log \Delta_H)$ rounds are sufficient to reduce the $(deg+1)$-list coloring problem on an $n$-node graph $H$  to several independent $(deg+1)$-list coloring instances on subgraphs with  $N=\polylog(n)$ nodes. In the remaining parts, we focus on handling one such subgraph. Let us detail the algorithm of \cite{BEPSv3}, to demonstrate its simplicity (the analysis is nontrivial though).
By \cite[Lemma 5.3]{BEPSv3}, after $O(\log\Delta_H)$ iterations of {\rct} in $H$, the maximum size of a connected component of uncolored nodes in $H$ is $O(\Delta_H^2\log n)$, w.h.p.\ One can now partition the uncolored vertices into two groups (depending on the uncolored degree) that need to be list-colored one after the other. In one of the groups, $O(\log\Delta_H)$ more iterations of {\rct} suffice to reduce the maximum size of an uncolored connected component to $N=\poly\log(n)$, while in the other group, uncolored components are bounded in size by $N$, by design. 
The analysis is independent of the colorspace, and this shattering procedure can be executed in \CONGEST. The two groups are processed sequentially, but inside each group, all uncolored components are colored in parallel. Notice that there are no conflicts between them, as all communication happens within individual components. Henceforth, we fix such a component $\smallComp$ and describe the coloring algorithm in $\smallComp$. By omitting excess colors, we can also assume that the list size of each vertex is upper bounded by $d_{\smallComp}(v)+1\leq |\smallComp|\leq N$. 

\paragraph{Network decomposition} 
To color a component $\smallComp$, we first compute a \emph{network decomposition} of $\smallComp$, using a deterministic algorithm from \cite{GGR20}. To highlight that all (uncolored) components of $G$ are handled in parallel and independently, we formulate all results in this section from the viewpoint of a graph $\smallComp$.
\begin{definition}
A \emph{network decomposition} of a graph $\smallComp$, with weak diameter $d$ and a number of colors $c$, consists of a partition of $\smallComp$ into vertex-induced subgraphs or \emph{clusters} $S_1,\dots,S_t$ and a coloring of the clusters with $c$ colors, such that: 1. For any $i\neq j$, if $S_i$ and $S_j$ have the same color, then there is no edge with one endpoint in $S_i$ and the other in $S_j$, 2. for each $i$ and $u,v\in S_i$, the distance from $u$ to $v$ in $\smallComp$ is bounded by $d$. \end{definition}
In a network decomposition with $c$ colors, let $\Gamma_1,\ldots,\Gamma_c$ denote the collection of clusters colored with the colors $1$ to $c$, respectively. 
\begin{theorem}\cite{GGR20}\label{thm:netdec}
Let $\smallComp$ be $k$-node graph where each node has a $b=\Omega(\log k)$-bit ID. There is a deterministic algorithm that computes a network decomposition of $\smallComp$ with $O(\log k)$ colors and weak diameter $O(\log^2 k)$, in $O(\log^5 k + (\log^* b) \log^4 k)$
rounds of the \CONGEST model with $b$-bit messages. Moreover, for each cluster $Q$, there is a Steiner tree $T_Q$ with radius $O(\log^2 k)$ in $\smallComp$,
for which the set of terminal nodes is $Q$. Each vertex of $\smallComp$ is in  $O(\log k)$ Steiner trees of each color.
\end{theorem}
Aggregation procedures can be efficiently pipelined in the computed network decomposition. 
\begin{lemma}[(simplified), \cite{GGR20}] \label{lem:pipelining}
 Let $\smallComp$ be a communication graph on n vertices. Suppose that each vertex of $\smallComp$ is
part of some cluster $Q$ such that each such cluster has a rooted Steiner tree $T_Q$  of diameter at most
$R$ and each node of $\smallComp$ is contained in at most $P$ such trees. Then, in $O(P + R)$ rounds of the
\CONGEST model with $b$-bit messages for $b \geq P$, we can perform the following operations for all
clusters in parallel: broadcast, convergecast, minimum, summation and bit-wise maximum.
\end{lemma}
For the precise definition of broadcast, convergecast, minimum and summation, we refer to \cite{GGR20}.
In the \emph{bit-wise maximum} of a cluster $Q$, each vertex $v\in Q$  has a bit string of $\ell=O(b)$ bits $s_1(v),\ldots,s_{\ell}(v)$, and we say $Q$ computes the bit-wise maximum if a designated leader node in $Q$ knows the bit string $s_1,\ldots,s_{\ell}$, where $s_i=\max_{v\in Q}s_i(v)$. 
While  the computation of  bit-wise maximum is not explicitly mentioned  in \cite{GGR20}, it follows from their claim that a function $\circ$ can be computed in the claimed runtime if it is associative and $p_i(x_1\circ \dots \circ x_k)$ can be computed from $p_i(x_1), \ldots, p_i(x_k)$ where $p_i$ denotes the $i$ leftmost (or rightmost) bits of a string. Both properties clearly hold for the bit-wise maximum.

We apply Theorem~\ref{thm:netdec}  to each component $\smallComp$. Since $\smallComp$ has size $O(\log^3 n)$ and IDs of size $O(n)$, we obtain a network decomposition with $O(\log\log n)$ colors and weak diameter $O(\log^2\log n)$, in $O(\log^5\log n)$ time, and each vertex is in at most $O(\log\log n)$ Steiner trees of each color.

\paragraph{Colorspace Reduction} We map the colorspace $[U]$ to a smaller colorspace that is comparable to the size $N$ of an uncolored component. 
 This is described in the following lemma, which is a reformulation of a similar lemma from~\cite{HKMN20}. 
 \begin{lemma}\label{lem:ColorSpaceReduction}
Let $N>\log n$ be an upper bound on the number of vertices and the diameter of some subgraph $H$ of a connected communication network $G$. Let the list of each node $v$ have size $|\Psi(v)|\le N$, and the colorspace have size $U=\poly(n)$. 
There is a constant $c_0>0$ such that all vertices of $H$ learn the same mapping $f:[U]\to [N^{c_0}]$ s.t. for each vertex $v\in H'$, $|f(\Psi(v))|=|\Psi(v)|$, by executing $O(\log N)$ iterative summation and broadcast primitives in $H$. 
\end{lemma}
\begin{proof}
Let the constant $c_0>0$ be such that $(N^{c_0}/2)^{N^{c_0-5}/2}>U$. Such a constant exists, since  $U=\poly(n)$ and $N>\log n$.
Let $N^{c_0}/2< p \le N^{c_0}$ be a prime, which exists due to Bertrand's postulate. The reduction is given by a random hash function, which is then derandomized using the method of conditional expectation. We start by describing this in the centralized setting, and then comment on how to do it distributively.

 Let $d=\lceil p/N^5\rceil$. To each color $\alpha\in [U]$, we assign a unique polynomial $\psi_\alpha\in \mathbb{F}_p[x]$ of degree $d$. This can be done since the number of degree-$d$ polynomials over $\mathbb{F}_p$ is $p^{d+1}>(N^{c_0}/2)^{N^{c_0-5}/2}>C$.

To obtain the color reduction $f:[U]\to [N^{c_0}]$ we want to choose  a ``good'' function from the set $\{f_g : g\in \mathbb{F}_p\}$, where 
\[f_g : [U] \rightarrow \mathbb{F}_p\text{, and } f_g(\alpha)=\psi_\alpha(g).\] 
In other words, we fix a distinct polynomial $\psi_{\alpha}$ for each color $\alpha\in [U]$ and for each $g\in \mathbb{F}$ we obtain a 'candidate color reduction' $f_g$ by mapping color $\alpha\in [U]$ to the evaluation of its polynomial $\psi_{\alpha}$ at $g$. The function $f_g$ is \emph{good} if it preserves the list sizes of all nodes $u$ of $H$: For each node $u$ and $g\in \mathbb{F}_p$, let $X_u(g)$ be 0 if  $|f_g(\Psi(u))|=|\Psi(u)|$, and 1, otherwise. A function $f_g$ is \emph{good} if $\sum_{u\in H} X_g(u)=0$; the latter implies that for each node $u$, $|f_g(\Psi(u))|=|\Psi(u)|$, as required. 

First, let us show that a random function $f_g$, corresponding to a uniformly random choice of $g$ is good with significant probability.
\begin{claim}
\label{lem:randReduction}
If $g\in \mathbb{F}_p$ is chosen uniformly at random, then $\E\big[\sum_{u\in  H} X_u\big]\le N^{-2}$.
\end{claim}
\begin{proof}
Consider a node $u$. For two distinct colors $\alpha,\beta\in \Psi(u)$, the polynomials $\psi_\alpha$ and $\psi_\beta$ intersect in at most $d$ points; hence, if $g$ is sampled uniformly, the probability that $f_g(\alpha)=\psi_\alpha(g)=\psi_\beta(g)=f_g(\beta)$ is $d/p\ge N^{-5}$. Thus, the probability that any two colors in $\Psi(u)$ map to the same element of $\mathbb{F}_p$ is at most $\binom{|\Psi(u)|}{2}\cdot d/p<N^2\cdot N^{-5}=N^{-3}$. 
The claim follows by linearity of expectation, since the sum is over at most $N$ elements, each having expectation at most $N^{-3}$.
\renewcommand{\qed}{\ensuremath{\hfill\blacksquare}}
\end{proof}
\renewcommand{\qed}{\hfill \ensuremath{\Box}}

We assume that the elements of $\mathbb{F}_p$ are numbered from 1 to $p$ and we choose the element $g$ by choosing its bits independently: we flip $\ell=\lceil\log_2 p\rceil$ unbiased coins,  $b_1,\dots,b_\ell$, and let $g$ be the number with the binary representation $b_1\dots b_\ell$. Let $Y$ be a binary random variable that is 1 if $g>p$, and 0, otherwise. Conditioned on $Y=1$, $g$ is uniformly distributed in $\mathbb{F}_p$.  Let us re-define $X_u(g)$ to be 0 also when $g>p$. Note that we still have $\E[\sum_{u\in H}X_u]= \E[\sum_{u\in H}X_u \mid Y=1]\Pr[Y=1]\le 1/N^2$, since we only make $X_u$ smaller in some cases. We also have $\E[Y]\le 1/2$; hence, assuming $N\ge 3$, 
\begin{equation}\label{eq:randomclred}
\E\left [Y+\sum_{u\in H}X_u(g)\right]\le 1/2+ 1/N^2\le 2/3\ .
\end{equation}
 In order to find a good function $f_g$, it suffices to find $b_1\dots b_\ell$ such that $Y+\sum_{u\in H}X_u<1$.
In order to derandomize (\ref{eq:randomclred}), we fix the bits $b_1,\dots,b_\ell$ inductively, with the basis (\ref{eq:randomclred}). Given $b_1,\dots,b_i\in \{0,1\}$, for $i\ge 0$, such that $\E[Y+\sum_{u\in H}X_u(g)\mid b_1,\dots,b_i]\le 2/3$, we select $b_{i+1}\in \{0,1\}$ such that $\E[Y+\sum_{u\in H}X_u(g)\mid b_1,\dots,b_i,b_{i+1}]\le 2/3$. The existence of such a value follows from the inductive hypothesis, using the conditional expectation formula (conditioning on $b_{i+1}$). After fixing all bits, we have a deterministic value $g$ for which $Y+\sum_{u\in H}X_u(g)\le 2/3$.

It remains to see how to compute a good function distributively. We do this in $\ell$ phases, where we fix $b_i$ in phase $i$, assuming the bits $b_1,\dots,b_{i-1}$ have been fixed, and each node knows those bits. To this end, each node computes $\E[X_u\mid b_1,\dots,b_{i-1},0]$ and $\E[X_u\mid b_1,\dots,b_{i-1},1]$, with precision $N^{-5}$. Note that such values fit in a single message of size $O(\log n)$. These values are aggregated in a leader node, i.e., a leader node learns values $v_b=\E\left[\sum_{u\in H}X_u\mid b_1,\dots,b_{i-1},b\right]\pm N^{-4}$, for $b=0,1$, and computes  $v_b'=\E\left[Y+\sum_{u\in H}X_u\mid b_1,\dots,b_{i-1},b\right]\pm N^{-4}$, and chooses $b_i=b$ such that $v_b'\le v'_{1-b}$. The overall error accumulated throughout the $\ell$ phases is at most $\ell N^{-4}\le N^{-3}$. This, implies that when all $\ell$ bits are fixed, we have $Y+\sum_{u\in H}X_u<2/3+N^{-3}<1$; hence, the computed function is good. It takes $O(\ell)=O(\log N)$ summation and broadcast operations to compute the color reduction.
\end{proof}

To obtain a color space reduction, i.e., a function $f_Q$, for each cluster $Q$ of the network decomposition, that maintains the list sizes, we apply the algorithm from Lemma~\ref{lem:ColorSpaceReduction} to all clusters in all $\Gamma_1,\ldots,\Gamma_c$ in parallel. Since we run it on  $O(\log\log n)$ color classes in parallel, and a node can be in $O(\log\log n)$ Steiner trees per color class, a node can be in $P=O(\log^2\log n)$ Steiner trees of all clusters. Using Lemma~\ref{lem:pipelining}, the $O(\log N)=O(\log\log n)$ iterative summation and broadcast primitives of Lemma~\ref{lem:ColorSpaceReduction} can be implemented in the Steiner trees of diameter at most $D=O(\log^2\log n)$ in $O(\log N\cdot (D+P))=O(\log^3\log n)$ rounds.

\paragraph{Cluster Coloring (by \cite{Ghaffari2019})}

 \begin{algorithm}[H]
\caption{Cluster Coloring}         
\label{alg:smalldeg}
\begin{algorithmic}[1]
\STATE{Given}: network decomposition $\Gamma_1,\ldots, \Gamma_c$, color space reduction $f_Q$, for each cluster $Q$
   \FOR{$i=1,\dots,c$}
      \FORALL{clusters $Q\in \Gamma_i$ in parallel}
				\STATE{\textbf{Update palettes:} Remove colors used by permanently colored neighbors (in $G$)}
							\STATE{\textbf{Map lists:}} Use $f_Q:[U]\rightarrow [N^{c_0}]$ to map remaining lists to a smaller color space 
				\FOR{$O(\log N)=O(\log\log n)$ iterations and $O(\log n)$ parallel instances}
				\STATE \textbf{Simulate} {\rct}
    		\ENDFOR
				\STATE \textbf{Agree on a successful instance} and permanently color each node in $Q$
    \ENDFOR
		\ENDFOR
\end{algorithmic}
\end{algorithm}
While the update of the lists takes place in the original color space, the simultaneous instances of {\rct} work in a smaller color space of size $\poly(N)$, and each of its colors can be represented with $O(\log N)=O(\log \log n)$ bits.  In fact, the color trials of one round of $O(\log n)$ instances use $O(\log n\cdot \log\log n)$ bits and can be done in $O(\log \log n)$ rounds. Thus, the $O(\log N)=O(\log\log n)$ rounds in total can be simulated in  $O(\log^2 \log n)$ rounds. Each of these instances is successful with probability $1/\poly(N)$ and w.h.p.\ (in $n$), at least one of them is successful~\cite[Sec. 4]{Ghaffari2019}. A node can locally determine which instances are successful, i.e., in which instances it gets colored. For node $v$, let $s_1(v)\cdots s_{\ell}(v)$ be the indicator string in which $s_i(v)$ indicates whether instance $i$ was successful for node $v$.  With a bit-wise maximum, a leader can determine an instance that is successful w.h.p.,\ for all nodes in the cluster, and can broadcast it to all nodes. 
Updating lists and applying the color space reduction can be done in one round, simulating the instances takes $O(\log^2\log n)$ rounds, and agreeing on a successful instance takes $O(\log^2\log n)$ rounds, due to the diameter of the Steiner trees and Lemma~\ref{lem:pipelining}. Thus, the total runtime by iterating over all color classes of the network decomposition is $O(\log^3\log n)$.

\begin{proof}[Proof of Theorem~\ref{thm:smallDegree}]
By design, after executing all four steps, all nodes of the graph are properly colored. 
The shattering part takes $O(\log \Delta_H)=O(\log \log n)$ rounds. Computing the network decomposition takes  $O(\log^5\log n)$ rounds.
The color space reduction runs on all clusters in parallel and takes $O(\log^3\log n)$ rounds. The cluster coloring part takes $O(\log^2\log n)$ rounds per color class of the network decomposition and $O(\log^3\log n)$ rounds in total. 
\end{proof}

\begin{remark}
A slightly worse runtime of $O(\log^6\log n)$ rounds can be obtained by using the algorithm of \cite{BKM19} to list color the clusters. Without the pipelining (\Cref{lem:pipelining}) to speed up aggregation within clusters  the runtime gets even slower but it remains $\poly\log\log(n)$. 
\end{remark}

\begin{remark}\label{rem:smallError}
The high level structure of the algorithm for \Cref{thm:smallDegree} is similar to the one in \cite{Ghaffari2019}, using the updated intermediate procedures to compute a network decomposition from \cite{GGR20} and the color reduction from \cite{HKMN20}. 

Unfortunately, the $O(\log \Delta)+2^{O(\log\log n)}$ algorithm in \cite{Ghaffari2019}  has a mistake in the design of its color space reduction.
To reduce the color space \cite{Ghaffari2019} maps the original color lists to a smaller color space $[p]$ where $p$ is a fixed (and deterministically chosen) prime in $[N^4,2N^4]$.  Color $x$ is mapped to $h_{a,b}(x)=a\cdot x + b\bmod p$ with randomly chosen $a$ and $b$. The wrong but crucial claim in the paper states that the probability for two colors $x$ and $x'$ to be hashed to the same color over the randomness of $a$ and $b$, i.e., the probability of the event that $h_{a,b}(x)=h_{a,b}(x')$ holds, is at most $1/N$. But colors $x$ and $x'$ are mapped to the same color whenever $x=x'\bmod p$, regardless of the choice of $a$ and $b$. In contrast (besides other changes in the design) in the core part of  \Cref{lem:ColorSpaceReduction} we fix one distinct polynomial $h_x(y)$ per color $x$ and evaluate it at a randomly chosen $y\in \mathbb{F}_y$ to obtain a color in the smaller color space $[p]$.
 The $O(\log\Delta +\log^6\log n)$ algorithm from \cite{GGR20} uses the techniques of \cite{Ghaffari2019} in a black box manner, including the erroneous color space reduction. 
\end{remark}
It is interesting that the color space dependence also plays a role for deterministic algorithms whose complexity is expressed as $f(\Delta)+O(\logstar n)$; see the discussion in \cite{MT20}  that compares the color space dependent results \cite{MT20} with the related color space independent results in \cite{FHK}.

%% file: app-bounds.tex
\section{Concentration Bounds}

We use the following variant of Chernoff bounds for dependent random variables, that is obtained, e.g., as a corollary of Lemma 1.8.7 and Thms. 1.10.1 and 1.10.5 in~\cite{Doerr2020}.

\begin{lemma}[Generalized Chernoff]\label{chernoff}
Let $X_1,\dots,X_r$ be binary random variables, and $X=\sum_i X_i$.
\begin{enumerate}
    \item If $Pr[X_i=1\mid X_1=x_1,\dots,X_{i-1}=x_{i-1}]\le q$, for all $i\in [r]$ and $x_1,\dots,x_{i-1}\in \{0,1\}$ with $Pr[X_1=x_1,\dots,X_r=x_{i-1}]>0$, then for any $\delta\in (0,1)$,
    \begin{equation}\label{eq:chernoffless}
    Pr[X\le(1-\delta)qr]\le \exp(-\delta^2qr/2)\ .
    \end{equation}
    \item If $Pr[X_i=1\mid X_1=x_1,\dots,X_{i-1}=x_{i-1}]\ge q$, for all $i\in [r]$ and $x_1,\dots,x_{i-1}\in \{0,1\}$ with $Pr[X_1=x_1,\dots,X_r=x_{i-1}]>0$, then for any $\delta>0$,
    \begin{equation}\label{eq:chernoffmore}
    Pr[X\ge(1+\delta)qr]\le \exp(-\min(\delta^2,\delta)qr/3)\ .
    \end{equation}
\end{enumerate}
\end{lemma}

We will often use the following simple corollary of  bounds (\ref{eq:chernoffless}-\ref{eq:chernoffmore}).
\begin{lemma}\label{lem:nbrhdpreserved}
Let $S$ be a randomly sampled subset of vertices in $G$, each sampled independently, with probability $q$. Let $\gamma\in (0,1)$ and $\alpha>0$ be such that $\gamma^2\alpha q\Delta\ge c\log n$, for a sufficiently large constant $c>1$. 
For every subset $T\subseteq V(G)$ of vertices of size $|T|=\alpha\Delta$, we have $(1-\gamma)q\alpha\Delta\le |T\cap S|\le (1+\gamma)q\alpha\Delta$, w.h.p. 
\end{lemma}

We will also use the following variant of Talagrand's inequality. A function $f(x_1, \dots, x_n)$ is called \emph{$c$-Lipschitz} iff changing any single $x_i$ can affect the value of
$f$ by at most $c$. Additionally, $f$ is called \emph{$r$-certifiable} iff whenever $f(x_1, \dots , x_n) \ge s$, there exists at
most $r s$ variables $x_{i_1}, \dots, x_{i_{rs}}$
so that knowing the values of these variables certifies $f\ge s$.
\begin{lemma}\label{lem:talagrand}\cite{molloy2013coloring}[Talagrand's Inequality II]
 Let $X_1, \dots , X_n$ be $n$ independent random variables and $f(X_1, \dots , X_n)$ be a $c$-Lipschitz $r$-certifiable function. 
For any $b \ge 1$,
\[
P (|f - \E[f]| > b+60c\sqrt{r\E[f]}) \le 4 \exp\left(-\frac{b^2}{8c^2r\E[f]}\right)\ .
\]
\end{lemma}